\documentclass[11pt,a4paper,oneside]{memoir}
\usepackage{amssymb}
\usepackage{amsthm}
\usepackage{mathtools}		
\usepackage{wrapfig}
\usepackage{geometry}
\usepackage[utf8]{inputenc} 	
\usepackage{bbm} 			
\usepackage{enumitem} 		
\usepackage{cleveref}		
\usepackage{etoolbox}    	
\usepackage{csquotes}		
\usepackage{authblk}
\usepackage{xcolor}
\usepackage{indentfirst}
\usepackage{ytableau}

\usepackage{tikz}
\usetikzlibrary{decorations.pathreplacing,
                calligraphy,
                matrix} 

\usepackage[style=ext-numeric-comp,articlein=false,maxbibnames=10,backend=biber,sorting=none,isbn=false,url=false,doi=false,giveninits=true,sortcites]{biblatex}
\AtEveryBibitem{%
  \clearlist{language}%
  \clearfield{number}%
}
\DeclareFieldFormat[article,periodical]{volume}{\mkbibbold{#1}}
\DeclareDelimFormat[bib,biblist]{nametitledelim}{.\space}
\DeclareDelimFormat[bib,biblist]{jourvoldelim}{,\space}

\usepackage{tikzsymbols}

\renewcommand{\S}{\mathbb{S}}

\DeclareMathOperator{\mult}{mult}

\DeclareMathOperator{\Sym}{Sym}

\DeclareMathOperator{\Cas}{\mathcal{Q}}

\DeclareMathOperator{\End}{End}
\DeclareMathOperator{\su}{\mathfrak{su}}
\let\sl\relax 
\DeclareMathOperator{\sl}{\mathfrak{sl}}
\DeclareMathOperator{\Stab}{Stab}

\DeclareMathOperator{\PStab}{Stab_{\mathbb{P}}}

\newcommand{\irrep}{\mathcal{H}}

\newcommand{\tensor}{\otimes}

\newcommand{\R}{\mathbb{R}}

\newcommand{\N}{\mathbb{N}}
\newcommand{\C}{\mathbb{C}}
\newcommand{\CP}{\mathbb{CP}}

\DeclarePairedDelimiter\angledbrackets\langle\rangle
\DeclarePairedDelimiter\curlybrackets\lbrace\rbrace
\DeclarePairedDelimiter\verticalline\lvert\rvert
\DeclarePairedDelimiter\doubleverticalline\lVert\rVert
\DeclarePairedDelimiter\floor\lfloor\rfloor

\DeclarePairedDelimiter\bra{\langle}{\rvert} 
\DeclarePairedDelimiter\ket{\rvert}{\rangle}

\newcommand{\abs}[1]{\verticalline*{#1}}

\newcommand{\set}[1]{\curlybrackets*{#1}}
\newcommand{\iset}[1]{\curlybrackets{#1}} 
\newcommand{\norm}[1]{\doubleverticalline*{\ifblank{#1}{{}\cdot{}}{#1}}}
\newcommand{\inorm}[1]{\doubleverticalline{\ifblank{#1}{{}\cdot{}}{#1}}}

\newcommand\braket[2]{\angledbrackets*{\ifblank{#1}{{}\cdot{}}{#1}\vert\ifblank{#2}{{}\cdot{}}{#2}}}

\DeclareMathOperator{\dx}{\,d\mathnormal{x}}
\DeclareMathOperator{\dw}{\,d\mathnormal{w}}

\DeclareMathOperator{\dz}{\,d\mathnormal{z}}

\DeclareMathOperator{\dt}{\,d\mathnormal{t}}

\DeclareMathOperator{\SU}{SU}

\DeclareMathOperator{\im}{im}

\DeclareMathOperator{\Tr}{Tr}

\theoremstyle{plain}
\newtheorem{definition}{Definition}[chapter] 
\newtheorem{theorem}[definition]{Theorem} 
\newtheorem{lemma}[definition]{Lemma}
\newtheorem{proposition}[definition]{Proposition}
\newtheorem{corollary}[definition]{Corollary}

\theoremstyle{definition}

\counterwithout{section}{chapter}
\counterwithout{definition}{chapter}
\numberwithin{equation}{section}
\setsecnumdepth{subsection}

\DeclareMathOperator{\Op}{Op\mathord{\mathrlap{\vphantom{\dagger}}}}
\DeclareMathOperator{\Hus}{H\mathord{\mathrlap{\vphantom{\dagger}}}}

\DeclareMathOperator{\U}{U}
\DeclareMathOperator{\Ber}{B}

\title{Optimal Remainder Estimates in the Quantization of~Complex Projective Spaces}
\author[1]{Tommaso Aschieri}
\author[2]{B\l a\.zej Ruba}
\author[1]{Jan Philip Solovej}
\affil[1]{Department of Mathematical Sciences, University of Copenhagen, \protect\\
Universitetsparken 5, 2100 Copenhagen, Denmark}
\affil[2]{Department of Mathematical Methods in Physics, \protect\\ Faculty of~Physics, University of Warsaw, \protect\\
 Pasteura 5, 02-093 Warszawa, Poland}
\date{\today}

\begin{document}

\maketitle
\begin{abstract}
We study Berezin–Toeplitz quantization of complex projective spaces $\CP^{d-1}$ and obtain full asymptotic expansions of the Berezin transformation and of products of Toeplitz operators. In each case, the remainder is controlled by the next term of the expansion, either through a positivity-preserving transformation or via an operator inequality. This leads to bounds which are optimal in terms of the required regularity and feature sharp or asymptotically sharp constants.
\end{abstract}
\section{Introduction}

In this paper, we study quantization, understood here as a correspondence between functions on a symplectic manifold and operators on certain Hilbert spaces. We address semiclassical approximation formulas for composition of operators with remainder bounds, optimal both in the regularity required of the quantized functions and in the constants appearing in the estimates. Such questions are  difficult in general, and in this paper we focus on the symplectic manifold $\CP^{d-1}$, the complex projective space, for which we are able to derive sharp estimates in the sense explained above. In fact, we do more, as a recurring theme in our results and proofs is positivity: for example, each remainder is related to the next term of the semiclassical expansion, either through a~positivity-preserving transformation or via an operator inequality.

In the setting we study, functions on $\CP^{d-1}$ are quantized as operators on the symmetric tensor spaces $\Sym^m(\C^d)$. The number $\frac{1}{m}$ plays the role of a semiclassical parameter. We have the familiar semiclassical correspondences: operator products and commutators approximate, respectively, multiplication and Poisson bracket of functions, up to correction terms of higher order in $\frac{1}{m}$. This in itself is by no means a~new result. Our contribution is to establish sharp estimates under relaxed regularity assumptions, which is crucial for the applicability of these tools in the quantitative analysis of operators.

Symmetric tensor spaces play an important role in mathematics: they may be viewed as spaces of complex homogeneous polynomials of degree $m$ in $d$ variables, as the $m$-particle subspace of the bosonic Fock space in quantum theory, and as irreducible representations of the compact group $\SU(d)$. On the geometric side, the complex projective space $\CP^{d-1}$ is of central importance in algebraic geometry, and is a~homogeneous space for $\SU(d)$.

The quantization of $\CP^{d-1}$ that we consider may be viewed as a specialized application of the Berezin-Toeplitz approach, in which Hilbert spaces arise as spaces of holomorphic sections of line bundles, and operators are given by compressions of multiplication operators. We~refer the reader to \cite{schlichenmaierBerezinToeplitzQuantizationCompact2010,le2018brief} and references therein for an~introduction to these methods. Here, we restrict the discussion to the specific case of our interest. 

In this Introduction, quantization will be introduced using coherent states \cite{Perelomov1972}. A~point $z \in \CP^{d-1}$  represents a line in $\C^d$, and its $m$-fold tensor product is a line in $\Sym^m(\C^d)$. We denote the corresponding orthogonal projection by $\ket{z}\bra{z}$ and refer to it as the coherent state. From the representation-theoretic perspective, coherent states form the $\SU(d)$-orbit of the highest weight projection. 

To every integrable function $f$ on $\CP^{d-1}$, we associate the linear operator
\begin{equation}
    \Op_m[f] = \dim(\Sym^m(\C^d)) \int_{\CP^{d-1}} f(z) \ket{z}\bra{z}\dz
\end{equation}
on $\Sym^m(\C^d)$. The function $f$ is called an upper symbol of the corresponding operator. Every linear operator $T$ on $\Sym^m(\C^d)$ admits a smooth upper symbol; that~is, there exists a smooth function $f$ such that $T=\Op_m[f]$. The lower symbol of $T$, also called the Husimi function, is defined by
\begin{align}
    \Hus_m[T](z) = \Tr[T\ket{z}\bra{z}]. 
\end{align}
The terminology of upper and lower symbols, proposed in \cite{simon_classical_1980}, is motivated by their utility in establishing operator bounds. For example, for $T = \Op_m[f]$ and $g = \Hus_m[T]$:
\begin{equation}
\inorm{g}_{L^\infty(\CP^{d-1})} \leq \| T \| \leq \| f \|_{L^\infty(\CP^{d-1})},
    \label{eq:lower_upper}
\end{equation}
where $\| T \|$ is the operator norm of $T$. 

The upper and lower symbols are related by the Berezin transformation
\begin{equation}
    \Ber_m := \Hus_m \Op_m.
\end{equation}
As already suggested by \eqref{eq:lower_upper}, the measure-theoretic properties of a function $f$ are well reflected in the spectral properties of $\Op_m[f]$ whenever $\Ber_m[f]-f$ is small. Moreover, the operator $\Ber_m-1$ quantifies how far the maps $\Op_m$ and $\Hus_m$ are from being inverses of each other. 

By general results in the Berezin-Toeplitz framework \cite{Englis,karabegovIdentificationBerezinToeplitzDeformation2001,ioosSpectralAspectsBerezin2020}, for every smooth function $f$ there exists an asymptotic expansion of $\Ber_m [f]$ as a formal power series in $\frac{1}{m}$ with coefficients given by universal differential operators acting on $f$. In the specialized case of $\CP^{d-1}$, our \Cref{thm:berezin_expansion} below improves upon these results. We show that the remainder after truncating the series is controlled by the corresponding next term; this demonstrates that the number of derivatives required is optimal, and that the explicit constants we obtain are asymptotically sharp, see \Cref{lemma:constants_asymptotically_optimal}.

Let us present a simplified version of \Cref{thm:berezin_expansion}, deferring the full statement to the main text. In preparation for that, let $m \in \mathbb N_0$. We introduce the entire function $\Upsilon_m$ and the sequence of coefficients $\upsilon_{m,n}$ by:
\begin{equation}
    \Upsilon_m(z) = \prod_{n=1}^\infty \left( 1 - \frac{z}{(m+n)(m+n+d-1)} \right) = \sum_{n=0}^\infty \upsilon_{m,n} (-z)^n.
\end{equation}
See also \eqref{eq:Upsilon_three_forms} below for alternative expressions for $\Upsilon_m(z)$. We note that $\upsilon_{m,n}$ is of order $m^{-n}$, as shown by the bounds $0 \leq \upsilon_{m,n} \leq \tfrac{m!}{(m+n)!}$. The function $\Upsilon_m$ allows to express \cite{berezin_general_1975,berezinQuantizationComplexSymmetric1975} the Berezin transformation in terms of the Laplacian $\Delta$ on $\CP^{d-1}$:
\begin{equation}
    \Ber_m = \Upsilon_m ( -\tfrac14 \Delta ).
\end{equation}

\begin{theorem}
Let $N \in \mathbb N$. Suppose that $f \in L^1(\CP^{d-1})$ is such that $\Delta^j f \in L^1(\CP^{d-1})$ for $1 \leq j \leq N$. For every $p \in [1,\infty]$ we have the bound:
\begin{align}
    \norm{\Ber_m [f]-\sum_{n=0}^{N-1} \upsilon_{m,n} \left(\tfrac{1}{4}\Delta\right)^n f}_{L^p(\CP^{d-1})}\leq \upsilon_{m,N}\norm{\left(\tfrac{1}{4}\Delta\right)^{N} f}_{L^p(\CP^{d-1})}.
\end{align}    
\end{theorem}

We remark that we actually have a stronger statement in \Cref{thm:berezin_expansion}, phrased in the language of majorization. This is explained in \Cref{sec:Op_defs} and \Cref{sec:berezin}. We have a~parallel result for maps $\Op_m \Hus_m$ in \Cref{thm:OpH}, which allows to compare operators with quantizations of their Husimi functions.

Let us now discuss products of operators $\Op_m[f]$. It is known from general results about quantization of compact K\"ahler manifolds   \cite{bordemannToeplitzQuantizationKahler1994,schlichenmaierZweiAnwendungenAlgebraischgeometrischer1996,schlichenmaierDeformationQuantizationCompact2000,maBerezinToeplitzQuantizationKahler2012} that if $f$ and $g$ are smooth functions, the product $\Op_m[f] \Op_m[g]$ admits an asymptotic expansion in terms of a~star product:
\begin{equation}
    \Op_m[f] \Op_m[g] \sim \Op_m[f \star g ] ,
    \label{eq:Op_product_star}
\end{equation}
where $\star$ is a formal series in $\frac{1}{m}$ of bidifferential operators, defining an associative product on $C^\infty(\CP^{d-1})[[ \tfrac1m ]]$. These results link \cite{karabegovIdentificationBerezinToeplitzDeformation2001} the Berezin-Toeplitz approach with deformation quantization \cite{bayenDeformationTheoryQuantization1978,bayenDeformationTheoryQuantization1978a,dewildeExistenceStarproductsFormal1983,fedosovSimpleGeometricalConstruction1994,kontsevichDeformationQuantizationPoisson2003} of K\"ahler manifolds \cite{karabegovDeformationQuantizationsSeparation1996,bordemannFedosovStarProduct1997,reshetikhinDeformationQuantizationKahler1999,gammelgaardUniversalFormulaDeformation2014,xuExplicitFormulaBerezin2012}. There has also been work \cite{barronSemiclassicalPropertiesBerezin2014,charlesSharpCorrespondencePrinciple2018,charlesEntanglementEntropyBerezin2020,aschieriSU2EquivariantQuantumChannels2024} on semiclassical approximations of operator products under restricted regularity assumptions, and with explicit constants in remainder estimates; however, these results do not achieve the level of precision offered by our methods. Finally, let us mention a~few selected references to recent work on exponential remainder estimates under analytic regularity assumptions \cite{deleporteToeplitzOperatorsAnalytic2021, charlesAnalyticBerezinToeplitz2021, roubyAnalyticBergmanOperators2020}.

We now present our result. Let $n \in \N_0$, and let $f,g$ be functions with square-integrable derivatives up to order~$n$. We define:
\begin{equation}
    f \star_n g = 2^{-n} \nabla^{(0,n)} f \cdot \nabla^{(n,0)} g,
    \label{eq:starn_def_in_intro}
\end{equation}
where $\nabla^{(0,n)}$ and $\nabla^{(n,0)}$ are, respectively, the $(0,n)$ and $(n,0)$ (anti-holomorphic and holomorphic) components of the $n$th Levi-Civita covariant derivative; see \Cref{sec:star} and \Cref{app:xi_appendix} for further discussion of the $\star_n$ products. Here we express the first two in elementary terms:
\begin{equation}
    f \star_0 g = fg, \qquad f \star_1 g = \tfrac14 \nabla f \cdot \nabla g - \tfrac{i}{4} \{ f , g \},
\end{equation}
where $\nabla f \cdot \nabla g$ is the Riemannian scalar product of gradients, and $\{ f , g \}$ denotes the Poisson bracket. 

We then define the star product of $f$ and $g$ by:
\begin{equation}
    f \star g = \sum_{n=0}^\infty \frac{(-1)^n (m+d-1)!}{n! (n+m+d-1)!}  f \star_n g.
    \label{eq:intro_star}
\end{equation}
Although this expression is not strictly a formal power series in $\frac{1}{m}$, the coefficient of the $n$th term is indeed of order $m^{-n}$.

In \Cref{sec:star} we discuss an algebra $\mathcal A$ of functions on $\CP^{d-1}$ for which the series \eqref{eq:intro_star} terminates after finitely many terms. In fact, $\star$ defines an associative product on $\mathcal A$ depending rationally on complex $m \not \in \{ - d , -d-1, \dots \}$; a similar construction can be found in \cite{cahenQuantizationKahlerManifolds1993}. For more general functions, we consider truncations of \eqref{eq:intro_star} at finite order and establish explicit bounds on the corresponding remainder in \eqref{eq:Op_product_star}.

\begin{theorem} \label{thm:2}
    Let $m, N \in \mathbb N_0$ and let $f,g$ be functions with square-integrable derivatives up to $N$th order. Let $\mathcal E_{m,N}[f,g]$ be the remainder of $N$th order approximation of $\Op_m[f] \Op_m[g]$ by $\Op_m[f \star g]$:
    \begin{equation}\label{eq:emn_error_term_def}
        \mathcal E_{m,N}[f,g] := \Op_m[f] \Op_m[g] - \sum_{n=0}^{N-1} \frac{(-1)^{n} (m+d-1)!}{n! (n+m+d-1)!} \Op_m[f \star_n g].
    \end{equation}
    We have the following remainder bounds:
    \begin{enumerate}[label=\arabic*)]
        \item In the case $g = \overline f$ we have the operator inequalities:
        \begin{align}\label{eq:remainder_bound_ffbar}
            0 \leq (-1)^N \mathcal E_{m,N}[f,\overline f]  \leq \frac{ (m+d-1)!}{N! (N+m+d-1)!} \Op_m[f \star_N \overline f].
        \end{align}
        \item There exist operators $A,B$ such that
    \begin{equation}
        \mathcal E_{m,N}[f,g] = \frac{(m+d-1)!}{N! (N+m+d-1)!} A^*B,
        \label{eq:AstarB}
    \end{equation}
    and we have the bounds
    \begin{equation}
        A^*A \leq \Op_m [f \star_N \overline f], \qquad B^* B \leq \Op_m [\overline g \star_N g]. 
        \label{eq:AstarA_BstarB}
    \end{equation}
    \end{enumerate}
\end{theorem}

The upper bounds in \eqref{eq:remainder_bound_ffbar} and \eqref{eq:AstarA_BstarB} are sharp, in the sense that equality is attained for certain functions in the algebra $\mathcal A$. These estimates lead to norm bounds of remainders; for instance, combining \eqref{eq:AstarB} with H\"older's inequality provides bounds for Schatten norms (see \Cref{cor:full_expansion_opf_opg} for details). The obtained bounds are also optimal in the required regularity: one derivative of $f$ and $g$ per order of the expansion, which matches the structure of the main terms.

Let us briefly mention some broader context related to our work. The orbit method \cite{kostantQuantizationUnitaryRepresentations1970,souriauStructureSystemesDynamiques1970,kirillovLecturesOrbitMethod} aims to relate irreducible representations of Lie groups to their coadjoint orbits; that is, group orbits in the dual space of the Lie algebra. For compact Lie groups, the coadjoint orbits are K\"ahler manifolds called flag manifolds. In this setting, the Borel-Weil theorem \cite{serreRepresentationsLineairesEspaces1954} realizes the program of the orbit method by expressing irreducible representations as spaces of holomorphic sections\footnote{Due to our choice of conventions, we actually work with anti-holomorphic sections in this paper.} of certain bundles on flag manifolds. An elementary example relevant for this paper is that homogeneous polynomials in $\Sym^m(\C^d)$ give rise to sections of bundles over complex projective spaces. One advantage of geometric constructions of representations, like in the Borel-Weil theorem, is that they provide a natural starting point for the application of semiclassical approximation methods in representation theory. Conversely, representation theory allows to perform certain computations exactly, leading to very precise results. A few examples of this interplay are \cite{guilleminGeometricQuantizationMultiplicities1982,guilleminGelfandCetlinSystemQuantization1983,landsmanStrictQuantizationCoadjoint1998,ioosAsymptoticsUnitaryMatrix2024,dawsonRepresentationtheoreticApproachToeplitz2025}.

Quantization of the complex projective space $\CP^{d-1}$ is also closely related to that of~$\C^d$, which leads to the Bargmann-Fock space $\bigoplus_{m=0}^\infty \Sym^m(\C^d)$ \cite{bargmannHilbertSpaceAnalytic1961}. The version of the Berezin-Toeplitz approach applicable in this setting gives rise to the anti-Wick ordering of creation and annihilation operators. The corresponding Berezin transformation, which is given by convolution with a Gaussian, relates it to the Wick ordering. The~complex projective space can be obtained from $\C^d$ as a symplectic quotient, also known as Marsden-Weinstein reduction. This was used in \cite{bordemannPhaseSpaceReduction1996} to derive a family of explicit star products on $\CP^{d-1}$. Let us informally explain a simple perspective on the reduction process. To every point $x \in \C^d$ corresponds the Glauber coherent state vector $\Omega_x$ in the Bargmann-Fock space \cite{klauderActionOptionFeynman1960,glauberQuantumTheoryOptical1963}. The projection of $\Omega_x$ onto the subspace $\Sym^m(\C^d)$ is proportional to the coherent state $\ket{z}$, where $z \in \CP^{d-1}$ is the line spanned by $x$. This allows to relate anti-Wick quantizations of certain functions to the operators $\Op_m[f]$ studied in this paper. With some loss of generality, let us consider a~function $F$ on $\C^d$ which is locally integrable and satisfies the homogeneity property: 
\begin{equation}
 F(\lambda z )  =|\lambda|^\nu F(z) \qquad \text{for } \lambda \in \C \setminus \{ 0 \}, \ z \in \C^d . 
\end{equation}
Then the restriction of $F$ to the unit sphere in $\C^d$ descends to a function $f$ on $\CP^{d-1}$. Let $\Op^{\mathrm{aW}}[F]$ be the anti-Wick quantization of $F$, see e.g.\ \cite[Ch.~2.7]{follandHarmonicAnalysisPhase1989} or \cite[Ch.~9.4]{derezinskiMathematicsQuantizationQuantum2023} for the definition. We have the following block-diagonal decomposition of $\Op^{\mathrm{aW}}[F]$:
\begin{equation}
    \Op^{\mathrm{aW}}[F] = \bigoplus_{m=0}^\infty \frac{\Gamma(m + d + \tfrac{\nu}{2})}{\Gamma(m+d)} \Op_m[f].
\end{equation}

The structure of this manuscript is as follows. \Cref{sec:preps}, which contains no new results, reviews some notions from the representation theory of $\SU(d)$ and the geometry of its homogeneous spaces that will be used throughout the paper. In \Cref{sec:Op_defs} we introduce the quantization maps under study, relating them to coherent state quantization and the general framework of Berezin–Toeplitz quantization. We also review spectral inequalities they satisfy using the language of majorization. \Cref{sec:berezin} is devoted to the asymptotic expansion of the Berezin transformations $\Hus_m\Op_m$ as well as the closely related maps $\Op_m\Hus_m$. Finally, \Cref{sec:star} concerns the expansion of the product $\Op_m[f]\Op_m[g]$. We review the properties of the operators $\star$ introduced in \eqref{eq:starn_def_in_intro} and prove \Cref{thm:2}.

\section*{Acknowledgments}

We would like to thank Jan Dereziński and Marcin Napiórkowski for discussions. The work of T.~A. and J.~P.~S. was supported by the Villum Centre of Excellence for the Mathematics of Quantum Theory (QMATH) with Grant No.10059. The work of B.~R. was supported by the National Science Centre (NCN) grant Sonata Bis 13 number 2023/50/E/ST1/00439.
\section{Preparations}\label{sec:preps}

\subsection{Lie algebra and representations} \label{sec:reps}

We will be working with the group $\SU(d)$, consisting of unitary $d \times d$ matrices with unit determinant. Its Lie algebra $\su(d)$ is the space of traceless anti-hermitian $d \times d$ matrices. The complexification of $\su(d)$ is the Lie algebra $\sl(d)=\sl(d,\C)$ of traceless $d \times d$ matrices. The standard basis of $\sl(d)$ is given by the matrix units $E_{i,j}$ with $i \neq j$, along with $d-1$ traceless diagonal matrices, e.g. $E_{i,i}-E_{i+1,i+1}$, $1 \leq i \leq d-1$. In a~slight abuse of notation, we will also use $E_{i,j}$ and $E_{i,i}-E_{i+1,i+1}$ to denote the action of these elements in various representations of $\sl(d)$. 

The standard Cartan subalgebra of $\sl(d)$ consists of the traceless diagonal matrices, and the standard choice of positive roots is such that $E_{i,j}$ with $i<j$ are root vectors for the positive roots. 

An important role will be played by the quadratic Casimir element, defined as an operator on any~representation of $\sl(d)$ by:
\begin{align}
    Q = \sum_{i \neq j}E_{i,j}E_{j,i} + \sum_{i=1}^{d} (E_{i,i} - \frac{1}{d} \sum_{k=1}^d E_{k,k})^2.
\end{align}
It acts as a scalar in irreducible representations of $\SU(d)$.

Let us introduce the irreducible representations of $\SU(d)$ that will be considered. We~have the tautological representation $\C^d$ and its dual space~$(\C^d)^*$, with the canonical basis $\ket{e_1},\dots,\ket{e_d} \in \C^d$ and the dual basis $\ket{e^1},\dots,\ket{e^d} \in (\C^d)^*$. The spaces $\C^d$ and $(\C^d)^*$ are isomorphic representations of $\SU(d)$ for $d=2$, but not for $d \geq 3$. We will also consider the spaces $\Sym^n (\C^d)$ and $\Sym^n (\C^d)^* \cong \Sym^n ((\C^d)^*) $ of symmetric tensors in $(\C^d)^{\tensor n}$ and $((\C^d)^*)^{\tensor n}$. These are irreducible representations of dimension
\begin{equation}
    d_n :=\binom{n+d-1}{n}.
\end{equation}

If $n,m \geq 1$, we have a homomorphism of representations: 
\begin{equation}\label{eq:tensor_contraction_definition}
   A  \ : \ \Sym^n(\C^d)^* \tensor \Sym^m(\C^d) \to     \Sym^{n-1}(\C^d)^* \tensor \Sym^{m-1}(\C^d)
\end{equation}
given by tensor contraction. It is surjective and its kernel is an irreducible representation \cite[p.~284]{knappLieGroupsIntroduction1996}, which we will denote by $\irrep_{n,m}$:
\begin{equation}
    \irrep_{n,m} := \{ \psi \in \Sym^n(\C^d)^* \tensor \Sym^m(\C^d) \, | \, A \psi =0  \}.
\end{equation}
We denote also $\irrep_{n,0} = \Sym^n(\C^d)^*$ and $\irrep_{0,m} = \Sym^m(\C^d)$. The representation $\irrep_{n,m}$ corresponds to the Young diagram:
\begin{figure}[htb]
    \centering
\begin{tikzpicture}[
BC/.style = {decorate, 
        decoration={calligraphic brace, amplitude=5pt, raise=1mm},
        very thick, pen colour={black}
            },
                    ]
\matrix (m) [matrix of math nodes,
             nodes={draw, minimum size=6mm, anchor=center},
             column sep=-\pgflinewidth,
             row sep=-\pgflinewidth
             ]
{
~   &~   & |[draw=none,text height=3mm]|\dots &   ~&  ~ &~ & |[draw=none,text height=3mm]| \dots &~  \\
~   &~   & |[draw=none,text height=3mm]|\dots  &   ~&   & & & \\
|[draw=none,text height=3mm]|\vdots &|[draw=none,text height=3mm]|\vdots   &   &   |[draw=none,text height=3mm]|\vdots&   & & & \\
~   &~   & |[draw=none,text height=3mm]|\dots  &   ~&   & &  &\\
};
\draw[BC] (m-4-1.south west) -- node[left =2.2mm] {$d-1$} (m-1-1.north west);
\draw[BC] (m-1-1.north west) -- node[above=2.2mm] {$n$} (m-1-4.north east);
\draw[BC] (m-1-5.north west) -- node[above=2.2mm] {$m$} (m-1-8.north east);
\end{tikzpicture}
\end{figure}

If $d \geq 3$, the representations $\irrep_{n,m}$ are pairwise non-isomorphic. The highest weight vector of $\irrep_{n,m}$ is $\ket{e^d}^{\otimes n} \otimes \ket{e_1}^{\otimes m}$. Acting on this vector one finds that on $\irrep_{n,m}$, the~Casimir $Q$ takes the value
\begin{equation}
    \left. Q \right|_{\irrep_{n,m} } = n^2 + m^2 - \frac{1}{d} (n-m)^2 + (d-1)(n+m).
    \label{eq:Casimir_values}
\end{equation}

The representations $\irrep_{n,m}$ appear in the tensor product decomposition
\begin{equation}
    \Sym^n(\C^d)^* \tensor \Sym^m(\C^d) \cong \bigoplus_{j=0}^{\min \{ n , m \}} \irrep_{n-j,m-j},
\end{equation}
where the sub-representation of $\Sym^n(\C^d)^* \tensor \Sym^m(\C^d)$ isomorphic to $\irrep_{n-j,m-j}$ can be characterized in terms of $A$ as follows:
\begin{equation}
\ker(A^{j+1}) \cap \ker(A^j)^\perp = \im((A^*)^j) \cap \im((A^*)^{j+1})^\perp \cong \irrep_{n-j,m-j}.
\end{equation}

Let us highlight an important special case of this decomposition:
\begin{equation}
    \End (\Sym^n(\C^d)) \cong \bigoplus_{j=0}^n \irrep_{j,j}.
\end{equation}
We denote by the different font $\Cas$ the Casimir element on the space of operators, on~which the Lie algebra acts by commutators:
\begin{align}
    \Cas[S] = QS + SQ &-2\sum_{i\neq j} E_{i,j} S E_{j,i} \label{eq:Casimir_adjoint} \\
    & -2 \sum_{i=1}^d (E_{i,i} - \tfrac1d\sum_{k=1}^d E_{k,k})S (E_{i,i} - \tfrac1d\sum_{k=1}^d E_{k,k}). \nonumber
\end{align}
\begin{lemma}\label{lemma:one_minus_casimir_QC}
The map:    
\begin{align}
    1-\Cas/\alpha \ \colon \End(\Sym^n(\C^d))\to \End(\Sym^n(\C^d))
\end{align}
is unital and trace preserving for every scalar $\alpha$. If $\alpha\geq \frac{2n(n+d)(d-1)}{d}$, then $1-\Cas/\alpha$ is completely positive.
\end{lemma}
\begin{proof}
We have $\Cas[1] =0$, so $1-\Cas/\alpha$ is unital. It is trace preserving because all commutators are traceless. By \eqref{eq:Casimir_adjoint}, in order for it to be completely positive it is enough that $S\mapsto \alpha S-QS-SQ$ is completely positive. Since $Q$ acts as a scalar \eqref{eq:Casimir_values} on $\Sym^n(\C^d)$, this is true exactly for $\alpha\geq \frac{2n(n+d)(d-1)}{d}$.
\end{proof}

\subsection{Homogeneous spaces} \label{sec:homog}

We denote the unit sphere in $\C^d$ by $\S^{2d-1}$. In terms of the standard linear coordinates $x_1,\dots,x_d$, it is the locus    $\sum_{i=1}^d |x_i|^2 =1$. It is a Riemannian manifold with the metric induced from the ambient Euclidean space $\C^d \cong \R^{2d}$. We normalize the associated Riemannian measure so that the total volume is $1$. The sphere $\S^{2d-1}$ can be identified with the homogeneous space $\SU(d)/\Stab(e_1)$, where $\Stab(e_1) \cong \SU(d-1)$ is the stabilizer of the vector $\ket{e_1} \in \C^d$. Explicitly,
\begin{align}
    \Stab(e_1) = \set{ \begin{bmatrix} 1 & 0 \\ 0 & V \end{bmatrix}, \ \ \ V \in \SU(d-1)}\subset \SU(d).
\end{align}

\begin{proposition} \label{prop:sphere_decomp}
    For every $n,m \in \N_0$ we have an isometric embedding $\irrep_{n,m} \to L^2(\S^{2d-1})$, which is the restriction to $\irrep_{n,m}$ of the map $((\C^d)^*)^{\otimes n} \otimes (\C^d)^{\otimes m} \to L^2(\S^{2d-1})$ given by
    \begin{equation}
      \ket{e^{i_1}} \otimes \cdots \otimes \ket{e^{i_n}} \otimes \ket{e_{j_1}} \otimes \cdots \otimes \ket{e_{j_m}}  \mapsto \sqrt{\frac{(d+n+m-1)!}{(d-1)!\,n!\,m!}} x_{i_1} \cdots x_{i_n} \overline x_{j_1} \cdots \overline x_{j_m}.
      \label{eq:Hpq_to_sphere}
    \end{equation}
    We have the orthogonal decomposition:
    \begin{equation}
        L^2(\S^{2d-1}) \cong \bigoplus_{n,m=0}^\infty \irrep_{n,m}.
        \label{eq:L2_sphere_decomp}
    \end{equation}
    On $\irrep_{n,m}$, the Laplacian on $\S^{2d-1}$ acts as multiplication by $-(n+m)(n+m+2d-2)$.
\end{proposition}
\begin{proof}
The map \eqref{eq:Hpq_to_sphere} is equivariant, hence isometric up to a factor by irreducibility of~$\irrep_{n,m}$. One can check it is isometric by calculating the norm of the right hand side for $i_1=\dots=i_n=1$, $j_1=\dots=j_m=d$. Such integrals can be calculated by comparing the integral over $\C^d$ of a gaussian multiplied by a polynomial, evaluated in Cartesian coordinates and in spherical coordinates. The direct sum
\begin{equation}
    \bigoplus_{n+m=k} \irrep_{n,m} \subset L^2(\S^{2d-1})
\end{equation}
is the set of restrictions of harmonic polynomials of degree $k$ to the sphere. It is easy to check that different $\irrep_{n,m}$ are orthogonal in $L^2(\S^{2d-1})$. The decomposition \eqref{eq:L2_sphere_decomp} then follows from Stone-Weierstrass, because every polynomial coincides on $\S^{2d-1}$ with some harmonic polynomial. The eigenvalue equation for the Laplacian is obtained by considering the Laplacian on $\C^d$ in spherical coordinates. 
\end{proof}

The complex projective space $\CP^{d-1}$ is the quotient of $\S^{2d-1}$ by the equivalence relation $x \sim e^{i \alpha} x$ for $\alpha \in \R$. The quotient carries an induced Riemannian structure, and~as in the case of the sphere we normalize its Riemannian measure to have total volume $1$. Alternatively, $\CP^{d-1}$ can be described as the quotient of $\C^d \setminus \{ 0 \}$ by the equivalence relation $x \sim e^{i \alpha} x$ for $\alpha \in \C$. This description endows $\CP^{d-1}$ with a~complex structure. The metric and the complex structure are compatible, making $\CP^{d-1}$ a K{\"a}hler manifold. The complex projective space $\CP^{d-1}$ can be identified with the homogeneous space $\SU(d)/\PStab(e_1)$, where $\PStab(e_1) \cong \U(d-1)$ is the stabilizer of the line spanned by the vector $\ket{e_1} \in \C^d$, or equivalently the stabilizer of the corresponding point in $\CP^{d-1}$. Explicitly,
\begin{align}
    \PStab(e_1) = \set{ \begin{bmatrix} \det(V)^{-1} & 0 \\ 0 & V \end{bmatrix}, \ \ \ V \in \U(d-1)}\subset \SU(d).
\end{align}

We will reserve the symbols $z,w$ to denote points of $\CP^{d-1}$.

One can identify functions on $\CP^{d-1}$ with functions on $\S^{2d-1}$ constant on the equivalence classes. Therefore, \Cref{prop:sphere_decomp} leads immediately to the decomposition of $L^2(\CP^{d-1})$ into irreducible representations of $\SU(d)$.

\begin{proposition}\label{prop:L2_proj_space_decomposition_and_casimir}
    We have the orthogonal decomposition:
    \begin{equation}
        L^2(\CP^{d-1}) \cong \bigoplus_{n=0}^\infty \irrep_{n,n}.
    \end{equation}
    On $\irrep_{n,n}$, the Laplacian on $\CP^{d-1}$ acts as multiplication by $-4n(n+d-1)$ and coincides with $-2Q$.
\end{proposition}

The group manifold $\SU(d)$ carries two commuting actions of $\SU(d)$, given by left and right translations. The corresponding actions on functions are given by:
\begin{align}
[\pi_L(h)f](g) = f(h^{-1}g),\qquad [\pi_R(h) f](g) = f(gh).
\end{align}
The Casimirs of these two representations coincide. 

Let $d \pi_L, d \pi_R$ be the differentials of $\pi_L$ and $\pi_R$, extended to the complexification $\sl(d)$ of $\su(d)$ by $\C$-linearity. We will denote 
\begin{equation}
\xi_{i,j} = - d\pi_R(E_{i,j}) \text{ for } i \neq j, \qquad \xi_{i,i} - \xi_{j,j} = - d\pi_R(E_{i,i} - E_{j,j}). 
\label{eq:xi_defs}
\end{equation}
Here the minus sign is included in order for $\xi_{i,j}$ to coincide with $d \pi_L(E_{i,j})$ at the neutral element of $\SU(d)$. As a consequence, the correspondence between the Lie algebra elements $E_{i,j}$ and the vector fields $\xi_{i,j}$ is a Lie algebra anti-homomorphism. We have:
\begin{subequations}
\begin{align}
    [\xi_{i,j},\xi_{k,l}] &= \delta_{i,l} \xi_{k,j} - \delta_{k,j} \xi_{i,l},  \\
    [\xi_{i,i}-\xi_{j,j} , \xi_{kl}] & = (\delta_{il} + \delta_{jk}  - \delta_{ik} - \delta_{jl}) \xi_{kl},  \\
 [\xi_{i,i}-\xi_{j,j} , \xi_{kk}-\xi_{ll}]&=0 ,
\end{align}\label{eq:xi_commutation}
\end{subequations}
where $i \neq j$ and $k \neq l$.

To make the above definitions more concrete, we now give an explicit formula for the action of the vector fields $\xi_{i,j}$ on matrix elements of the group action. If $\irrep$ is a~finite-dimensional representation of $\SU(d)$, and $\ket{\psi},\ket\varphi\in \irrep$, then:
\begin{align}\label{eq:xi_action_on_functions}
    \xi_{i,j} \bra{\psi} g^{-1} \ket{\varphi} = \bra\psi E_{i,j} g^{-1}\ket\varphi.
\end{align}

We highlight that the vector fields $\xi_{i,j}$, being invariant under the left translation action but not under the right translations, do not descend to smooth vector fields on the homogeneous spaces $\S^{2d-1}$ and $\CP^{d-1}$. Nevertheless, they can be used in the differential calculus on homogeneous spaces by identifying functions on homogeneous spaces with functions on the group that satisfy certain symmetry properties. For example, the $L^2$ space on the sphere can be identified with
\begin{equation}
    L^2(\S^{2d-1}) \cong \{ f \in L^2(\SU(d)) \, | \, f(gh)=f(g) \text{ for all } h \in \Stab(e_1)  \}.
\end{equation}
Similarly, functions on $\CP^{d-1}$ can be identified with functions on $\SU(d)$ invariant under right translation by elements of $\PStab(e_1)$. We will need a slight generalization of this. The map
\begin{equation}
    \mu   \colon \PStab(e_1)\ni \begin{bmatrix}
        \det(V)^{-1} & 0 \\ 
        0 & V 
    \end{bmatrix} \mapsto \det(V)^{-1} \in \S^1. 
\end{equation}
is a group homomorphism. We define, for every $m \in \N_0$,
\begin{equation}
    L^2(\CP^{d-1},m) = \{ f \in L^2(\SU(d)) \, | \, f(gh) = \mu(h)^{-m} f(g) \text{ for all } h \in \PStab(e_1) \}.
\end{equation}
Functions in $L^2(\CP^{d-1},m)$ can be identified with square integrable functions on $\S^{2d-1}$ that satisfy $f(e^{i \alpha} x) = e^{- i m \alpha} f(x)$ for $\alpha \in \R$, or equivalently with sections of a suitable holomorphic line bundle on $\CP^{d-1}$ (scalar-valued functions for $m=0$). Since we will use the differential operators $\xi_{i,j}$ heavily, it is most convenient to regard elements of $L^2(\CP^{d-1},m)$ as functions on the group manifold. The relation between geometric structures on $\CP^{d-1}$ and the vector fields $\xi_{i,j}$ is discussed in \Cref{app:xi_appendix}.  

\begin{proposition}\label{prop:sections_decomp}
We have the orthogonal decomposition:
\begin{equation}
    L^2(\CP^{d-1},m) \cong \bigoplus_{n=0}^\infty \irrep_{n,n+m}.
    \label{eq:L2CPm}
\end{equation}
An isometric, equivariant embedding $V_m : \Sym^{m}(\C^d) \to L^2(\CP^{d-1},m)$, whose image is the $n=0$ summand $\irrep_{0,m}$ in \eqref{eq:L2CPm}, is given by
\begin{equation}
    V_m \ket \psi(g) = d_m^{\frac12}  \bra{e_1}^{\otimes m} g^{-1} \ket{\psi}. 
    \label{eq:Vm_def}
\end{equation}
Moreover, for every $2\leq i\leq d$:
\begin{align}\label{eq:xi_ij_annihilate_hw_vector}
    \xi_{i,1} V_m\ket\psi = 0.
\end{align}
\end{proposition}
\begin{proof}
    The orthogonal decomposition follows from \Cref{prop:sphere_decomp}, and \eqref{eq:Vm_def} from Schur's orthogonality relations. Finally, \eqref{eq:xi_ij_annihilate_hw_vector} is a consequence of \eqref{eq:xi_action_on_functions}.
\end{proof}

The matrix element in \eqref{eq:Vm_def} is the scalar product of $\ket{\psi}$ with the vector $g \ket{e_1}^{\otimes m}$, which is often called the coherent state.  In the interpretation of $L^2(\CP^{d-1},m)$ as the space of $L^2$ sections over~$\CP^{d-1}$, the image of $V_m$ becomes the space of anti-holomorphic sections. This identification is a special case of the celebrated Borel-Weil theorem. We note that $V_m^* V_m^{\vphantom{*}}=1$, while $V_m^{\vphantom{*}} V_m^*$ is the orthogonal projection onto the image of $V_m$.

\section{Quantization and symbol maps} \label{sec:Op_defs}

Let us introduce the quantization maps that will be studied in the next sections. Following the definition of Perelomov \cite{Perelomov1972}, we consider the coherent states:
\begin{align}\label{eq:coherent_state_perelomov}
    \iset{g\ket{e_1}^{\tensor m}}_{g\in \SU(d)}\subseteq \Sym^m(\C^d).
\end{align}
They form an overcomplete set of vectors. The complex line spanned by $g \ket{e_1}^{\otimes m}$ is invariant under multiplication of $g$ by an element of $\PStab(e_1)$; in other words, it depends only on the coset $g \PStab(e_1) \in \SU(d)/\PStab(e_1) = \CP^{d-1}$. We denote by $\ket{z}\bra{z}$ the rank one projection in $\Sym^m(\C^d)$ corresponding to $z=g \PStab(e_1) \in \CP^{d-1}$. Then $\ket{z}\bra{z}$ is a~smooth operator-valued function on $\CP^{d-1}$. We note that vectors $g \ket{e_1}^{\tensor m}$ with different $g$ in the same coset $g \PStab(e_1)$ differ by a complex phase. 

With a slight abuse of notation, we sometimes write $\ket{z}$ for an some arbitrarily chosen unit vector in the range of the projection $\ket{z} \bra{z}$. If $m \neq 0$, it is not possible to choose $\ket{z}$ so that it is a~continuous vector-valued function on $\CP^{d-1}$.

Coherent states can be used to define the quantization map $\Op_m$. On functions in $L^1(\CP^{d-1})$ it is given by an absolutely convergent integral:
\begin{align}
    \Op_m[f] = d_m \int_{\CP^{d-1}} f(z) \ket{z}\bra{z}\dz \in \End(\Sym^m(\C^d)),
\end{align}
where $d_m = \dim(\Sym^m(\C^d))$. One can extend the map $\Op_m$ to distributions on $\mathbb{CP}^{d-1}$ by interpreting the integral suitably. An alternative, equivalent expression for $\Op_m[f]$ is
\begin{equation}
    \Op_m[f] = V_m^* f V_m^{\vphantom{*}},
    \label{eq:Op_Toeplitz}
\end{equation}
where $V_m$ is the isometric embedding in \eqref{eq:Vm_def}, and $f$ acts as a multiplication operator. Therefore, if one identifies $\Sym^m(\C^d)$ with the image of $V_m$, then $\Op_m[f]$ becomes the compression of a multiplication operator, also called a Toeplitz operator.

Equally important is the Husimi (or symbol) map $\Hus_m$:
\begin{align}
    \Hus_m[T](z) = \bra{z}T \ket{z} = \Tr(T \ket{z} \bra{z}) ,
\end{align}
which sends operators in $\End(\Sym^m(\C^d))$ to smooth functions on $\CP^{d-1}$. $\Hus_m$ is, up to an overall normalization, the adjoint of $\Op_m$.

The following properties of $\Op_m, \Hus_m$ are easy to verify, mostly using Schur's lemma.

\begin{lemma} \label{lem:oph_elementary}
Let $m\in\N_0$. The maps $\Op_m$, $\Hus_m$ satisfy the following properties.
\begin{enumerate}[label=\arabic*)]
    \item $\Op_m$ takes positive functions to positive operators, and $\Hus_m$ takes positive operators to positive functions. 
    \item We have $\Op_m[\overline f] = \Op_m[f]^*$ and $\Hus_m[T^*] = \overline{\Hus_m[f]}$.
    \item We have $\Op_m[1]=1$, $\Hus_m[1]=1$, and
    \begin{subequations}
    \begin{align}
      \frac{1}{d_m}  \Tr(\Op_m[f]) &= \int_{\CP^{d-1}} f \dz,\\
      \int_{\CP^{d-1}} \Hus_m[T] \dz &= \frac{1}{d_m} \Tr (T).
    \end{align}
    \end{subequations}
\end{enumerate}
\end{lemma}

Let us discuss some spectral properties of $\Op_m, \Hus_m$. Firstly, we have the celebrated Berezin-Lieb inequalities \cite{lieb_classical_1973, berezin_general_1975, simon_classical_1980}. For any convex function $\varphi$, real-valued~$f$, and self-adjoint $T$, one has the bounds:
\begin{subequations}
\begin{align}\label{eq:berezin_lieb_inequalities}
     \frac{1}{d_m} & \Tr[\varphi(\Op_m[f])] \leq \int_{\CP^{d-1}} \varphi(f)\dz, \\
     \int_{\CP^{d-1}} \varphi(\Hus_m[T]) \dz \leq  \frac{1}{d_m} & \Tr[\varphi(T)].
\end{align}
\end{subequations}

It is useful to reformulate these inequalities in the language of majorization. For a~measurable, real-valued function $f$ on a finite measure space $(X,\mu)$ its (signed) decreasing rearrangement $f^*$ is the unique decreasing, right-continuous function on $[0,\mu(X))$ that is equimeasurable with $f$, that is:
\begin{align}
    \mu(\iset{f(x) \in I}) = \abs{\iset{f^*(t)\in I}} \quad \textup{for any closed, bounded interval }I.
\end{align}
If $\mu(X)=\infty$, one can define $f^*$ analogously for non-negative, integrable functions \cite{chongExtensionsTheoremHardy1974}.

One can check that $f^*$ is given by the formula:
\begin{align}
        f^*(t) = \inf\iset{s \in \R \colon \mu(\iset{f(x)>s})\leq t}.
\end{align}
We remark that this definition of $f^*$ is slightly non-standard. It is more common (\cite{alvinoOptimizationProblemsPrescribed1989, grafakosClassicalFourierAnalysis2014,bennettInterpolationOperators1988}) to define $f^*$ as the unique decreasing, right-continuous function on $[0,\mu(X))$ that is equimeasurable with $\abs{f}$. In other words, $f^*$ is what we denote by $\abs{f}^*$. 

Let $f$ and $g$ be real-valued measurable functions on $(X,\mu)$ and $(X',\mu')$ with integrable positive parts. Assume that $\mu(X)=\mu'(X') <\infty$. One says that $g$ weakly majorizes $f$, written $f\prec_w g$, if for all $s\in [0,\mu(X))$:
\begin{align}
\int_{0}^s{} f^*(t)\dt &\leq \int_0^s g^*(t)\dt,\label{eq:def_majorization_inequality}
\shortintertext{where finiteness of the integrals is ensured by the assumptions on $f,g$. If $f\prec_w g$ and:}
\quad\int_{0}^{\mu(X)}f^*(t)\dt &= \int_{0}^{\mu(X)}g^*(t)\dt,\label{eq:def_majorization_equality}
\end{align}
then one says that $g$ majorizes $f$, written $f\prec_w g$.

We remark that this notion of majorization is the generalization of a corresponding notion for vectors in $\mathbb R^n$ \cite[Section 2.18]{hardyInequalities1923}, for which the integrals are replaced by finite sums.

The link between majorization and the Berezin-Lieb inequalities is provided by the following version of Karamata's theorem, see \cite[Theorem 2.5]{chongExtensionsTheoremHardy1974}.
\begin{proposition}\label{prop:karamata_for_functions}
Let $f\in L^1(X, \mu)$, and $g\in L^1(X',\mu')$ be real-valued functions, and assume that $\mu(X)=\mu(X')<\infty$. Then $f\prec g$ if and only if for any convex function $\varphi$:
\begin{align}
    \int_X \varphi(f)\dx \leq \int_X \varphi(g)\dx.
    \label{eq:majorization_Jensen}
\end{align}
Under the same assumptions, $f\prec_w g$ if and only if 
\eqref{eq:majorization_Jensen} holds for every convex function $\varphi\colon \R\to \R_{\geq 0}$ satisfying $\lim_{x \to - \infty} \varphi(x)=0$.
\end{proposition}

Let us make the connection with Berezin-Lieb inequalities explicit. If $T$ is a self-adjoint operator on a finite-dimensional Hilbert space $\irrep$, the sequence of its eigenvalues $\iset{\lambda_i(T)}$ can be regarded as the function $\Lambda(T)$ on the set $\iset{1,\ldots, \dim\irrep}$, given by 
\begin{equation}
 \Lambda(T)(i) = \lambda_i(T), \qquad \lambda_1(T) \geq \lambda_2(T) \geq \dots   .
\end{equation}
Equipping $\iset{1,\ldots, \dim\irrep}$ with the counting measure normalized to have total mass $1$, one finds that:
\begin{align}
    \Lambda(T)^* = \sum_{i=1}^{\dim\irrep} \lambda_{i}(T)\mathbbm{1}_{[(i-1)/\dim\irrep,i/\dim\irrep)}.
\end{align}
For later use, we mention that for a pair of self-adjoint operators $T,S$ acting on vector spaces of the same dimension, majorization $\Lambda(T) \prec \Lambda(S)$, often denoted by $T \prec S$, is~equivalent to the eigenvalue inequalities
\begin{equation}
    \sum_{i=1}^k \lambda_i(T) \leq \sum_{i=1}^k \lambda_i(S),
\end{equation}
with equality for $k = \dim \irrep$.

In light of the discussion above, the Berezin-Lieb inequalities are equivalent to the majorization statements:
\begin{align}
    \Hus_m[T] \prec \Lambda(T),\qquad \Lambda(\Op_m[f])\prec f, \label{eq:HusOp_majorization}
\end{align}
for any real-valued integrable $f$ and self-adjoint $T$. More explicitly, denoting by $\floor{ x}  $ and $\iset{x}=x-\lfloor x \rfloor$ the integer and the fractional part of $x \in \mathbb R$, for every $s \in [0,1)$ we have:
\begin{subequations}
\begin{align}
    \int_{0}^s\Hus_m[T]^*\dt \leq \frac{1}{d_m}\sum_{i=1}^{\lfloor sd_m\rfloor}\lambda_i(T) + \frac{\iset{sd_m}}{d_m}\lambda_{\lfloor sd_m\rfloor+1}(T) ,\\
    \frac{1}{d_m}\sum_{i=1}^{\lfloor sd_m\rfloor}\lambda_i(\Op_m[T]) + \frac{\iset{sd_m}}{d_m}\lambda_{\lfloor sd_m\rfloor+1}(\Op_m[T]) \leq \int_0^s f^*\dt,
    \end{align}
\end{subequations}
with equalities for $s=1$.

Let us present a direct proof of the majorization statements \eqref{eq:HusOp_majorization}.

\begin{proof}[Proof of \eqref{eq:HusOp_majorization}]
Equalities for $s=1$ are included in the elementary \Cref{lem:oph_elementary}. We verify weak majorization.

Let $f\in L^1(\CP^{d-1})$ be real valued. By the variational principle:
\begin{align}\label{eq:first_bl_inequality_proof}
    \int_{0}^s \Lambda(\Op_m[f])^* dt = \frac{1}{d_m}\sup_{\substack{0\leq S \leq 1\\ \Tr S \leq s d_m}} \Tr(\Op_m[f]S) = \sup_{\substack{0\leq S\leq 1\\ \Tr S\leq s d_m}} \int_{\CP^{d-1}} f \Hus_m[S]\dz.
\end{align}    
The Hardy-Littlewood inequality (see \cite[Theorem 4.2]{chongSpectralInequalitiesInvolving1982}) gives:
\begin{align}
    \sup_{\substack{0\leq S\leq 1\\ \Tr S\leq s d_m}} \int_{\CP^{d-1}} f \Hus_m[S]\dz \leq \sup_{\substack{0\leq S\leq 1\\ \Tr S\leq s d_m}} \int_{0}^1 f^* \Hus_m[S]^*\dt.
\end{align}
One can easily check that $\int \Hus_m[S] = \frac{1}{d_m}\Tr[S]\leq s$, and $0\leq \Hus_m[S]\leq 1$, so that the integral is upper bounded by the integral in which $\Hus_m[S]^*$ is replaced by $\mathbbm{1}_{[0,s]}$. This gives:
\begin{align}
\int_{0}^s \Lambda(\Op_m[f])^* dt\leq \int_0^s f^*\dt,    \label{eq:last_bl_inequality_proof}
\end{align}
which is the desired weak majorization result.

For the other direction, let $T\in \End(\Sym^m(\C^d))$. We see that, for $0\leq s < 1$:
\begin{align}
    \int_0^s \Hus_m[T]^* \dt = \sup_{\substack{\abs{A}=s}} \int_{A} \Hus_m[T]\dz = \frac{1}{d_m}\sup_{{\abs{A}=s}} \Tr(\Op_m[\mathbbm{1}_A] T).
\end{align}
We get a larger quantity if we maximize over a larger set of positive operators:
\begin{align}
    \frac{1}{d_m}\sup_{{\abs{A}=s}} \Tr(\Op_m[\mathbbm{1}_A] T)\leq \frac{1}{d_m}\sup_{\substack{0\leq S \leq 1\\\Tr S\leq sd_m} } \Tr(S T) = \int_{0}^s T^* \dt,
\end{align}
where the equality follows from the variational principle.
\end{proof}

Majorization results for eigenvalues lead to the following bounds on singular values, valid also for non-hermitian operators. 

\begin{lemma}\label{lemma:weak_majorization_bl_inequalities}
    Let $f\in L^1(\CP^{d-1})$ and $T\in \End(\Sym^m(\C^d))$. Then:
    \begin{align}\label{eq:weak_majorization_bl_inequalities}
        \abs{\Hus_m[T]}\prec_w \Lambda(\abs{T}),\qquad \Lambda(\abs{\Op_m[f]})\prec_w \abs{f},
    \end{align}
    where $\abs{T}=\sqrt{T^*T}$ is the absolute value of $T$.
\end{lemma}
\begin{proof}

Let $T\in \End(\Sym^m(\C^d))$. By the polar decomposition, we write $T=U\abs{T}$, where $U$ is a unitary operator. It follows, by the Cauchy-Schwarz inequality, that:
\begin{align}
\abs{\Hus_m[T](z)} = \abs{\bra{z}T\ket{z}} &\leq \inorm{\sqrt{\abs{T}}U^*\ket{z}}^{1/2}\inorm{\sqrt{\abs{T}}\ket{z}}^{1/2}\\
&\leq \frac{1}{2}\bra{z}U\abs{T}U^*\ket{z}+ \frac{1}{2} \bra{z}\abs{T}\ket{z}.\nonumber   
\end{align}
By \eqref{eq:HusOp_majorization}, both terms in the last inequality are majorized by $\Lambda(U\abs{T}U^*) = \Lambda(\abs{T})$. Thus:
\begin{align}
    \abs{\Hus_m[T]}\prec_w \Lambda[\abs{T}].
\end{align}

The proof for $f\in L^1(\CP^{d-1})$ is essentially identical to the one given in \eqref{eq:first_bl_inequality_proof}-\eqref{eq:last_bl_inequality_proof}, with the following modifications. The supremum is replaced by:
\begin{align}
    \int_{0}^s \Lambda(\abs{\Op_m[f]})^* \dt = \frac{1}{d_m}\sup_{\substack{\abs{S}\leq 1\\ \Tr\abs{S}=sd_m}}\abs{\Tr[\Op_m[f]S]}.
\end{align}
One needs to apply the Hardy-Littlewood inequality to $\int \abs{f\Hus_m[S]}$, and then use that  $\int \abs{\Hus_m[S]} \leq \tfrac{1}{d_m}\Tr\abs{S}\leq s$, which holds because $\abs{\Hus_m[S]}\prec_w \Lambda(\abs{S})$.
\end{proof}

Let us now state bounds on $\Op_m, \Hus_m$ in Schatten and $L^p$ spaces. The argument we present avoids the use of interpolation methods employed e.g.\ in \cite{simon_classical_1980} to prove the same inequalities.
\begin{lemma}
Let $m\in\N_0$. We have the bounds
    \begin{subequations}
    \begin{align}
        \| \Op_m[f]  \|_{\mathcal S^p}   &\leq d_m^{  \frac1p} \| f \|_{L^p(\CP^{d-1})}, \label{eq:norm_bound_op}\\ 
        \| \Hus_m[T] \|_{L^p(\CP^{d-1})} &\leq d_m^{- \frac1p} \| T \|_{\mathcal S^p},   \label{eq:norm_bound_hus}
    \end{align}
    \label{eq:norm_bound}
    \end{subequations}
    where $\| \cdot \|_{\mathcal S^p}$ is the $p$th Schatten norm on operators on $\Sym^m(\C^d)$.
\end{lemma}
\begin{proof}
  The norm bounds for $p<\infty$ follow from \Cref{prop:karamata_for_functions} and \Cref{lemma:weak_majorization_bl_inequalities} because the function $x \mapsto x^p$ is convex. The case $p=\infty$ can be obtained by taking a limit.
\end{proof}

We remark that the norms of spaces $L^p(\CP^{d-1})$ and $S^p$ appearing in the norm bounds \eqref{eq:norm_bound} can be represented by the same function norm $\rho_p(\cdot) = \| \cdot \|_{L^p([0,1))}$ on $[0,1)$:
\begin{subequations}
\begin{align}
    \norm{f}_{L^p(\CP^{d-1})}&=\rho_p(\abs{f}^*) = \left(\int_{0}^1 (\abs{f}^*)^p\dx\right)^{1/p},\\
    d_m^{-\frac1p}\norm{T}_{\mathcal{S}^p} &= \rho_p(\Lambda(\abs{T})^*) = \left(\int_0^1 (\Lambda(\abs{T})^*)^p\dx\right)^{1/p}.
\end{align} 
\end{subequations}
Similar bounds can be obtained for more general rearrangement invariant norms $\rho$ on functions \cite{luxemburgRearrangementInvariantBanach1967} on $\CP^{d-1}$, and $\rho'$ on operators on $\End(\Sym^m(\C^d))$, if they both can be represented by the same function norm $\rho''$ on $[0,1)$:
\begin{align}
\rho(f) = \rho''(|f|^*), \quad \rho'(T) = \rho''(\Lambda(|T|)^*).
\end{align}
See for example \cite[Theorem 4.6]{bennettInterpolationOperators1988} for a reference. Examples of such norm functionals include the Lorentz maximal norms $\norm{}^*_{p,q}$ for $1<p<\infty$.
\section{The Berezin transformation} \label{sec:berezin}

For every $m \in \N_0$ we consider the \textit{Berezin transformation}
\begin{equation}
    \Ber_m = \Hus_m \Op_m,
\end{equation}
which is a linear map from distributions to smooth functions on $\CP^{d-1}$. It satisfies 
\begin{equation}
    \Ber_m[1] =1, \qquad \int_{\CP^{d-1}} \Ber_m[f]\dz = \int_{\CP^{d-1}} f\dz,
\end{equation}
preserves positivity, and is a positive operator on $L^2(\CP^{d-1})$. It is given by an explicit integral kernel:
\begin{equation}
    \Ber_m[f](z) = d_m \int_{\CP^{d-1}} |\braket{z}{w}|^{2m} f(w) \dw .
    \label{eq:Ber_kernel}
\end{equation}

\begin{proposition}
For every $m\in \N_0$:
\begin{align}
    \Ber_m = \sum_{n=0}^m \frac{m!(m+d-1)!}{(m-n)!(m+d-1+n)!}\Pi_{n,n},
    \label{eq:Berq_spectral_decomp}
 \end{align}
where $\Pi_{n,n}$ denotes the orthogonal projection onto $\irrep_{n,n}$ in $L^2(\CP^{d-1})$.
\end{proposition}
\begin{proof}
$B_m$ acts as a scalar in each $\irrep_{n,n}$. To compute the corresponding eigenvalue, we~pick a function $f$ in $\irrep_{n,n}$ using \Cref{prop:sphere_decomp} and compute the integral \eqref{eq:Ber_kernel}.
\end{proof}

We remark that $\Ber_m$ and $\Pi_{n,n}$ are operators with smooth integral kernels, and \eqref{eq:Berq_spectral_decomp} can be interpreted as equality of operators from the space of distributions to smooth functions.

\begin{corollary} \label{cor:Berezin_convergence}
For every $p \in [1,\infty)$ we have that $\Ber_m \to 1$ on $L^p(\CP^{d-1})$ strongly as $m \to \infty$, and $\Ber_m \to 1$ strongly on $C(\CP^{d-1})$. The convergence is monotone (in the sense of operator inequalities) on $L^2(\CP^{d-1})$. 
\end{corollary}
\begin{proof}
Every coefficient in \eqref{eq:Berq_spectral_decomp} converges to $1$ as $m \to \infty$, so we have convergence on the dense subspace $\bigcup_{m=0}^\infty \bigoplus_{n=0}^m \mathcal H_{n,n}$. The convergence is monotone because the coefficients in \eqref{eq:Berq_spectral_decomp} are increasing with $m$.
\end{proof}

Let $m \in \N_0$. We consider the entire function 
\begin{subequations}
\begin{align}
    \Upsilon_m(z) &= \prod_{n=1}^\infty \left( 1 - \frac{z}{(m+n)(m+n+d-1)}  \right) \label{eq:Berezin_product} \\
    &= \frac{\Gamma(m+d) \Gamma(m+1)}{\Gamma \left( m + \frac{d+1}{2} + \sqrt{\left( \frac{d-1}{2} \right)^2 + z } \right) \Gamma \left( m + \frac{d+1}{2} - \sqrt{\left( \frac{d-1}{2} \right)^2 + z } \right)} \label{eq:lambdaq_Gammas} \\
    &=\frac{m!(m+d-1)!}{\pi}\frac{\sin\left(\pi\Big(m+\tfrac{d+1}{2}+\sqrt{\left( \frac{d-1}{2} \right)^2 + z }\Big)\right)}{{\displaystyle\prod_{j=1}^{2m+d}}\left(m+\tfrac{d+1}{2}-j-\sqrt{\left( \frac{d-1}{2} \right)^2 + z }\right)}. 
\end{align}
\label{eq:Upsilon_three_forms}
\end{subequations}
The equality of the first two lines can be verified using Weierstrass' product formula for the Gamma function (see also \cite[(4), p.5]{batemanHigherTranscendentalFunctions1953}). The equality of the last two lines follows from Euler's reflection formula and the functional equation satisfied by the Gamma function. We highlight that \eqref{eq:Berezin_product} is the product of affine linear functions of $z$ which vanish at the eigenvalues of $-\tfrac14 \Delta$, except the first $m+1$ eigenvalues (multiplicities ignored).

\begin{corollary}\label{corollary:ber_holomorphic_function_lp}
The Berezin transformation $\Ber_m$ acting on $L^2(\CP^{d-1})$ is expressed using the functional calculus of the Laplacian as:
\begin{equation}
    \Ber_m = \Upsilon_m( -\tfrac14 \Delta ).
\end{equation}
\end{corollary}

We remark that since $\Upsilon_m$ is a holomorphic function, \Cref{corollary:ber_holomorphic_function_lp} allows to obtain $\Ber_m$ using complex integration along suitably chosen paths.


We denote the coefficients of the Taylor expansion of $\Upsilon_m(-z)$ at zero by $\upsilon_{m,n}$:
\begin{equation}    
\Upsilon_m(z) = \sum_{n=0}^\infty \upsilon_{m,n} (-z)^n.
\end{equation}
The coefficients $\upsilon_{m,n}$ are all positive and given by
\begin{equation}
    \upsilon_{m,n} = \sum_{\substack{I \subset \N \\ |I|=n}} \prod_{k \in I}  \frac{1}{(m+k)(m+d-1+k)}.
\end{equation}
One can obtain closed form expressions for these sums by differentiating \eqref{eq:lambdaq_Gammas} at zero and using the well-known properties of the Gamma function and the polygamma functions. The first few are:
\begin{subequations}
\begin{align}
    \upsilon_{m,0} =& 1, \\
    \upsilon_{m,1} =&  \frac{1}{d-1} \sum_{j=1}^{d-1} \frac{1}{m+j}, \\
    \upsilon_{m,2} =& \frac{1}{2} \upsilon_{m,1}^2 + \frac{1}{(d-1)^2} \upsilon_{m,1} + \frac{1}{(d-1)^2} \sum_{j=1}^m \frac{1}{j^2}  \\
    & + \frac{1}{(d-1)^2} \sum_{j=1}^{d-1} \frac{1}{(m+j)^2} - \frac{\pi^2}{6(d-1)^2}. \nonumber
\end{align}
\end{subequations}
The coefficient $\upsilon_{m,1}$ is rational, but $\upsilon_{m,2}$ is transcendental. One can verify that for every $n \geq 2$ the coefficient $\upsilon_{m,n}$ belongs to the subring of $\R$ generated over $\mathbb Q$ by $\zeta(2), \dots, \zeta(n)$, where $\zeta$ is the Riemann zeta function. 

We note the useful bound:
\begin{equation}
    \upsilon_{m,n} \leq \prod_{j=0}^{n-1} \upsilon_{m+j,1} \leq \frac{m!}{(m+n)!}.
\end{equation}
In particular, in the limit $m \to \infty$ with fixed $n$ we have $\upsilon_{m,n} = O(m^{-n})$.

\begin{theorem} \label{thm:berezin_expansion}
Let $N\in \N $. Suppose that $f \in L^1(\CP^{d-1})$ is such that $\Delta^j f \in L^1(\CP^{d-1})$ for $1 \leq j \leq N$. If $f$ is real-valued, then we have majorization:
\begin{equation}
    \Ber_m [f]-\sum_{n=0}^{N-1} \upsilon_{m,n} \left(\tfrac{1}{4}\Delta\right)^n f \prec \upsilon_{m,N} (-\tfrac14\Delta)^N f.
\end{equation}
For complex-valued $f$, we have the weak majorization:
\begin{equation}
     \left| \Ber_m [f]-\sum_{n=0}^{N-1} \upsilon_{m,n} \left(\tfrac{1}{4}\Delta\right)^n f \right| \prec_w \upsilon_{m,N} |(\tfrac14 \Delta)^N f|. 
\end{equation}
In particular, for every $p \in[1,\infty]$ we have the bound:
\begin{align}\label{eq:Ber_expansion_bound_functions}
    \norm{\Ber_m [f]-\sum_{n=0}^{N-1} \upsilon_{m,n} \left(\tfrac{1}{4}\Delta\right)^n f}_{L^p(\CP^{d-1})}\leq \upsilon_{m,N}\norm{\left(\tfrac{1}{4}\Delta\right)^{N} f}_{L^p(\CP^{d-1})}.
\end{align}    
\end{theorem}
\begin{proof}
The following identities hold on the space of distributions:
\begin{equation}
    \Ber_m = \left( 1 + \frac{\tfrac14 \Delta}{(m+1)(m+d)} \right) \Ber_{m+1}, \qquad \Ber_m \Delta = \Delta \Ber_m.
\end{equation}
Hence by induction, for every $n \in \mathbb N$:
\begin{equation}
    \Ber_m - \Ber_{m+n} = \sum_{k=1}^n \frac{\Ber_{m+k}}{(m+k)(m+d-1+k)} \cdot \tfrac14\Delta .
\end{equation}
Let us take $f \in L^1(\CP^{d-1})$ such that $\Delta f\in L^1(\CP^{d-1})$. Corollary \ref{cor:Berezin_convergence} then gives, in the limit $n\to\infty$:
\begin{equation}
\label{eq:Berf_minus_f}
    \Ber_m[f] - f = \sum_{k=1}^\infty \frac{1}{(m+k)(m+d-1+k)}  \Ber_{m+k}[\tfrac14\Delta f].
\end{equation}
We highlight that the series converges absolutely in $L^1(\CP^{d-1})$. Each $\Ber_{m+k}$ is a doubly stochastic operator (see \cite{dayDecreasingRearrangementsDoubly1973}) and $\sum_{k=1}^\infty \frac{1}{(m+k)(m+d-1+k)} = \upsilon_{m,1}$, so 
\begin{equation}
    \Ber_m[f]-f = \upsilon_{m,1} T_1 [\tfrac14\Delta f ]
    \label{eq:ds_majorizer}
\end{equation}
for some doubly stochastic operator $T_1$. More precisely, $T_1$ has integral kernel $K_1(z,w)$ which is a non-negative measurable function satisfying
\begin{equation}
    \int K_1(z,w) \dz = \int K_1(z,w) \dw = 1.
\end{equation}
Hence, $\Ber_m[f] -f \prec \upsilon_{m,1} \tfrac14\Delta f$ if $f$ is real-valued, see e.g.\ \cite[Theorems~4.1~and~4.9]{dayDecreasingRearrangementsDoubly1973}. If~$f$ is complex-valued, \eqref{eq:ds_majorizer} implies
\begin{equation}
   | \Ber_m[f] - f | \leq \upsilon_{m,1} T_1[|\tfrac14\Delta f| ] \prec \upsilon_{m,1} |\tfrac14\Delta f|. 
\end{equation}

Now let us suppose that $f \in L^1(\CP^{d-1})$ is such that $\Delta^j f \in L^1(\CP^{d-1})$ for $j \leq N$. An inductive argument using \eqref{eq:Berf_minus_f} repeatedly gives:
\begin{align}\label{eq:induction_step_berezin_inversion}
        \Ber_m[f] - \sum_{n=0}^{N-1}\upsilon_{m,n} (\tfrac14\Delta)^n f&= \smashoperator{\sum_{\substack{I \subset \N \\ |I|  ={N}}}}\phantom{1,}\prod_{k\in I}\frac{1}{(m+k)(m+d-1+k)} \Ber_{m+\max I}[ (\tfrac14\Delta)^{N} f], \nonumber \\
        & = \upsilon_{m,N} T_N [(\tfrac14 \Delta)^N f], 
\end{align}
where $T_N$ is a doubly stochastic operator. We proceed as earlier for $N=1$.

The norm bounds follow from weak majorization, since $x\mapsto x^p$ is convex.
\end{proof}

Since remainder terms in \Cref{thm:berezin_expansion} are bounded by the subsequent expansion terms, the constants in our estimates are asymptotically optimal. We spell out a precise argument in \Cref{lemma:constants_asymptotically_optimal} below. We remark that the constants are not strictly optimal, at least for $p=2$, for which the operator norm is expressed in terms of the eigenvalues. The suboptimal step is to use that the doubly stochastic operator $T_N$ in \eqref{eq:induction_step_berezin_inversion} is a~contraction: while it is true that $\| T_N \|=1$, the norm of $T_N$ is slightly smaller when restricted to the orthogonal complement of constant functions (the image of the Laplace operator).

\begin{lemma}\label{lemma:constants_asymptotically_optimal}
    Let $N \in \N$, $m \in \N_0$ and $p \in [1,\infty]$. Suppose that $c >0$ is the optimal constant such that for every smooth function $f $ we have
    \begin{equation}
    \Big \| \Ber_m [f]-\sum_{n=0}^{N-1} \upsilon_{m,n} \left(\tfrac{1}{4}\Delta\right)^n f  \Big \|_{L^p(\CP^{d-1})} \leq c \| (\tfrac14 \Delta)^N f \|_{L^p(\CP^{d-1})}.
\label{eq:optimal_Ber_inequality}
    \end{equation}
    Then $c$ satisfies the bounds:
    \begin{equation}
\upsilon_{m,N}\left(1-\frac{d}{N+m+1}\right) \leq c \leq \upsilon_{m,N}.
    \end{equation}
\end{lemma}
\begin{proof}
The upper bound on $c$ follows from \Cref{thm:berezin_expansion}. To prove the lower bound, recall that the lowest nonzero eigenvalue of $-\frac14 \Delta$ is $d$. Choose a corresponding eigenvector $f \neq 0$. The left hand side of inequality \eqref{eq:optimal_Ber_inequality} is lower bounded by:
\begin{align}
\upsilon_{m,N}\norm{(\tfrac{1}{4}\Delta)^{N+1} f}-\norm{\Ber_m[f]-\sum_{n=0}^{N}\upsilon_{m,n}(\tfrac{1}{4}\Delta)^n f} \geq \upsilon_{m,N}d^N - \upsilon_{m,N+1}d^{N+1}.
\end{align}
The result follows by comparing with the right hand side of \eqref{eq:optimal_Ber_inequality} and noting that:
\begin{align}
\upsilon_{m,N+1}\leq \frac{1}{m+N+1}\upsilon_{m,N}.
\end{align}
\end{proof}

We have finished the analysis of the Berezin transformations $\Ber_m = \Hus_m \Op_m$, which act on functions. As we saw, it was useful to compare $\Ber_m$ with $\Ber_{m'}$, $m \neq m'$. Now we want to analyze maps $\Op_m \Hus_m$, which act on $m$-dependent vector spaces of operators. Instead of comparing $\Op_m \Hus_m$ to $\Op_{m'} \Hus_{m'}$, it will be useful to consider the maps $\Upsilon_{m'}(\tfrac12 \Cas)$ with varying $m'$ but acting on the same vector space $\End(\Sym^m(\C^d))$. We~will also need a~majorization statement for $2$-positive maps.

\begin{lemma} \label{lem:2pos_majorization}
Let $\Phi\colon \End(\irrep)\to \End(\irrep)$ be a positive map which is trace non-increasing and satisfies $\Phi(\mathbbm{1})\leq \mathbbm{1}$. Then, if $T$ is self adjoint:
\begin{align}
    \Phi(T)\prec_w T.
    \label{eq:weak_majorization_channel_without_abs}
\end{align}
If $\Phi$ is trace preserving, the majorization is strict. Furthermore, if $\Phi$ is $2$-positive:
\begin{align}\label{eq:weak_majorization_channel}
    \abs{\Phi(T)}\prec_w \abs{T},
\end{align}
for any $T\in \End(\irrep)$.
\end{lemma}

\begin{proof}
That $\Phi(T)\prec T$ is a simple consequence of the variational principle. Note that:
\begin{align}
    \int_{0}^{n/\dim\irrep} \Lambda(\Phi(T))^*\dt = \sup_{\dim W \leq n} \Tr[\Phi(T)P_W] = \sup_{\dim W \leq n} \Tr[T\Phi^*(P_W)], 
\end{align}
where $P_W$ is the orthogonal projection onto subspace $W$. Since $0\leq \Phi^*(P_W)\leq 1$ and $\Tr[\Phi^*(P_W)]\leq n$, it is possible to upper bound this quantity by:
\begin{align}
    \sup_{\substack{0\leq S \leq 1\\ \Tr[S]\leq n}} \Tr[TS] = \int_{0}^{n/\dim\irrep} \Lambda(T)^*\dt.
\end{align}
Equality for $n=\dim\irrep$ is true exactly when $\Tr[\Phi(T)]=\Tr[T]$.

Let now $T$ be any operator, and assume that $\Phi$ is $2$-positive. It suffices to prove that
\begin{align}\label{eq:weaker_majorization_for_2_positive_maps}
    \Lambda(\abs {\Phi(T)})\prec_w \frac{1}{2}\Lambda(\Phi(\sqrt{TT^*}))+\frac{1}{2}\Lambda(\Phi(\sqrt{T^*T})),
\end{align}
and note that by \eqref{eq:weak_majorization_channel_without_abs} both terms are weakly majorized by $\Lambda(\sqrt{TT^*})= \Lambda(\sqrt{T^*T})$. The proof of \eqref{eq:weaker_majorization_for_2_positive_maps} combines the ideas of \cite{audenaertArakiLiebThirringInequality2007, audenaertNoteP-qNorms2009}. We consider the self-adjoint operator
\begin{align}
    A = \begin{bmatrix}0&T\\T^*&0\end{bmatrix},
\end{align}
and write $A=A_+-A_-$, where $A_+, A_-$ are the positive and negative parts of $A$, so that 
\begin{equation}
 |A| = A_+ + A_- =   \begin{bmatrix}
     \sqrt{TT^* } & 0 \\ 0 & \sqrt{T^*T}
 \end{bmatrix}. 
\end{equation}
We see that:
\begin{subequations}
\begin{align}
    (\Phi\tensor \mathbbm{1}_2)(A) &= (\Phi\tensor \mathbbm{1}_2)(A_+) - (\Phi\tensor \mathbbm{1}_2)(A_-)\label{eq:2positive_dilation_1}\\
    (\Phi\tensor \mathbbm{1}_2)(\abs{A}) &=(\Phi\tensor \mathbbm{1}_2)(A_+) + (\Phi\tensor \mathbbm{1}_2)(A_-),\label{eq:2positive_dilation_2}
\end{align}    
\end{subequations}
where the individual terms are again positive, since $\Phi$ is $2$-positive. We now use the general fact that, for positive operators $X,Y$ and any rearrangement-invariant norm $\norm{}_*$ \cite{audenaertArakiLiebThirringInequality2007}:
\begin{align}
    \norm{X-Y}_*\leq \norm{X+Y}_*.
\end{align}
We apply this to \eqref{eq:2positive_dilation_1} and \eqref{eq:2positive_dilation_2} and consider the Fan norms $\norm{}_n$, given by:
\begin{align}
    \norm{S}_n = \sum_{i=1}^n \lambda_i(\abs{S}).
\end{align}
This immediately gives the claimed \eqref{eq:weaker_majorization_for_2_positive_maps}.
\end{proof}

\begin{theorem} \label{thm:OpH}
    \begin{enumerate}[label=\arabic*)]
        \item For every $m \in \N_0$ we have:
        \begin{equation}\label{eq:Op_Hus_diagonalization}
            \Op_m \Hus_m = \sum_{n=0}^m \frac{m!(m+d-1)!}{(m-n)!(m+d-1+n)!}\Pi_{n,n},
        \end{equation}
        where $\Pi_{n,n}$ is the orthogonal projection onto $\irrep_{n,n}$ in $\End(\Sym^m(\C^d))$.
        \item We have the identity:
        \begin{equation}\label{eq:Op_Hus_holomorphic_function_casimir}
            \Op_m \Hus_m = \Upsilon_m( \tfrac12 \Cas).
        \end{equation}
        \item Let $N \in \N$ and let $T \in \End(\Sym^m(\C^d))$. If $T$ is self-adjoint, we have the majorization:
        \begin{equation}
         \Op_m \Hus_m [T]-\sum_{n=0}^{N-1} \upsilon_{m,n} \left(-\tfrac{1}{2}\Cas \right)^n [T]  \prec  \upsilon_{m,N}\left(-\tfrac{1}{2}\Cas \right)^{N} [T], 
        \end{equation}
        and more generally, for arbitrary $T$:
        \begin{equation}
            |\Op_m \Hus_m [T]-\sum_{n=0}^{N-1} \upsilon_{m,n} \left(-\tfrac{1}{2}\Cas \right)^n [T]|  \prec_w  \upsilon_{m,N} |\left(-\tfrac{1}{2}\Cas \right)^{N} [T]|, 
        \end{equation}
        In particular, we have the Schatten norm bounds:
        \begin{align}\label{eq:Op_Hus_expansion_bound}
    \norm{\Op_m \Hus_m [T]-\sum_{n=0}^{N-1} \upsilon_{m,n} \left(\tfrac{1}{2}\Cas \right)^n [T]}_{\mathcal S^p} \leq \upsilon_{m,N}\norm{\left(\tfrac{1}{2}\Cas \right)^{N} [T]}_{\mathcal S^p}.
\end{align}   
    \end{enumerate}
\end{theorem}
\begin{proof}
    Since $\Hus_m$ is, up to a scalar factor, the adjoint of $\Op_m$, it follows that $\Op_m\Hus_m$ has the same eigenvalues as $\Ber_m=\Hus_m\Op_m$. Since both maps are equivariant, \eqref{eq:Op_Hus_diagonalization} follows. That \eqref{eq:Op_Hus_holomorphic_function_casimir} holds is a consequence of this and \eqref{eq:lambdaq_Gammas}, since $\tfrac12\Cas$ acts as the scalar $n(n+d-1)$ on $\irrep_{n,n}$.

    Finally we prove \eqref{eq:Op_Hus_expansion_bound}. Let us introduce the linear maps on $\End(\Sym^m(\C^d))$:
    \begin{align}\label{eq:generalized_ber_definition}
        \tilde\Ber_k = \Upsilon_{k}(\tfrac12\Cas) = \prod_{n=1}^\infty \left( 1 - \frac{\tfrac12\Cas}{(k+n)(k+n+d-1)}  \right), \quad k\in\N_0,\ k\geq m.
    \end{align}
    In particular, $\tilde\Ber_m = \Op_m\Hus_m$. Each factor in the product \eqref{eq:generalized_ber_definition} is of the form $1-\Cas/\alpha$, with $\alpha \geq 2m(m+d)(d-1)/d$. By \Cref{lemma:one_minus_casimir_QC}, every such factor is a unital quantum channel. It follows that $\tilde\Ber_k$ is a unital quantum channel for every $k\geq m$. 
    
    Finally, it is clear that $\tilde\Ber_k\to1$ as $k\to\infty$, and the proof of \eqref{eq:Op_Hus_expansion_bound} now follows using the same strategy as that of \eqref{eq:Ber_expansion_bound_functions}, simply replacing $-\tfrac14\Delta$ by $\tfrac12\Cas$. One gets the identity:
\begin{align}\label{eq:induction_step_berezin_inversion_on_operators}
        \tilde\Ber_m[S] - \sum_{n=0}^{N-1}\upsilon_{m,n} (-\tfrac12\Cas)^nS &= \smashoperator{\sum_{\substack{I \subset \N \\ |I|  ={N}}}}\phantom{1,}\prod_{k\in I}\frac{1}{(m+k)(m+d-1+k)} \tilde\Ber_{m+\max I}[ (-\tfrac12\Cas)^{N}S], \nonumber \\ 
        & = \upsilon_{m,N}\tilde{T}_N [(-\tfrac12 \Cas)^N S]. 
    \end{align}
    By the above argument, $\tilde{T}_N$ is a unital quantum channel. We apply Lemma \ref{lem:2pos_majorization} and conclude the desired majorizations.
\end{proof}
\section{Products of Toeplitz operators} \label{sec:star}

In this section we discuss how to approximate operator products $\Op_m[f] \Op_m[g]$ for $m \to \infty$. As a preparation to state our main theorem, we introduce the terms of the expansion. Let $f,g$ be functions on $\CP^{d-1}$ with square-integrable derivatives up to $n$th order, and let:
\begin{align}
    f\star_n g := (-1)^n \sum_{i_1,\ldots, i_n=2}^d  (\xi_{1,i_1}\cdots \xi_{1,i_n} f) (\xi_{i_n,1}\cdots \xi_{i_1,1}g),
    \label{eq:starn_def}
\end{align}
where $\xi_{i,j}$ are the left-invariant vector fields on $\SU(d)$, defined in \eqref{eq:xi_defs}. This formula should be interpreted as follows: we regard $f,g$ as $\PStab(e_1)$-invariant functions on $\SU(d)$, and we compute derivatives on the group manifold. The right hand side of \eqref{eq:starn_def} (but not the individual terms of the summation) is also $\PStab(e_1)$-invariant, and hence can be regarded as a function on $\CP^{d-1}$. The map $(f,g)\mapsto f\star_n g$ is an invariant bi-differential operator of degree $n$ in both $f$ and $g$. As we will see, the definition \eqref{eq:starn_def} is equivalent to \eqref{eq:starn_def_in_intro}.

Let us state a few properties of the $\star_n$ operators.

\begin{proposition} \label{prop:starn_properties}
The operators $\star_n$ have the following properties.
    \begin{enumerate}[label=\arabic*)]
        \item Positivity: $\overline f \star_n f \geq 0$.
        \item The Hermitian form defined by $\star_n$ is given by a function of the Laplacian:
        \begin{equation}
            \int_{\CP^{d-1}} (\overline f \star_n g)(z) \dz = \braket{f }{p_n(- \tfrac14 \Delta) g}, 
            \label{eq:star_n_integral}
        \end{equation}
        where $p_n(z) = \prod_{j=0}^{n-1} (z-j(j+d-1))$ is the monic polynomial of degree $n$ vanishing on the first $n$ eigenvalues of $- \frac14 \Delta$. 
        \item If $f$ or $g$ is in the subspace $\bigoplus_{j=0}^{n-1} \irrep_{j,j} \subset L^2(\CP^{d-1})$, then $f \star_n g=0$.
        \item Description of $f \star_n g$ without lifting to $\SU(d)$: let $\nabla^{(0,n)}$, $\nabla^{(n,0)}$ be the $(0,n)$ and $(n,0)$ (i.e.\ anti-holomorphic and holomorphic) components of the $n$th (Levi-Civita) covariant derivative on $\CP^{d-1}$. Then $\nabla^{(0,n)}f$, $\nabla^{(n,0)}g$ are symmetric tensors, and 
        \begin{equation}
            f \star_n g = 2^{-n} \, \nabla^{(0,n)} f \cdot \nabla^{(n,0)} g,
        \end{equation}
        where $\cdot$ denotes $n$-fold contraction using the Riemannian metric tensor.
        \item The first two $\star_n$ products are
        \begin{equation}
            f \star_0 g = fg, \qquad f \star_1 g = \frac{1}{4} \left( \nabla f \cdot \nabla g - i \{ f ,g \} \right),
        \end{equation}
        where $\nabla f \cdot \nabla g$ is the Riemannian contraction of gradients of $f,g$ and $\{ f,g \}$ is the Poisson bracket of $f,g$.
    \end{enumerate}
\end{proposition}
\begin{proof}
1) We have $\xi_{1,i} = - \overline{\xi_{i,1}}$, so
\begin{equation}
    \overline f \star_n f = \sum_{i_1,\dots,i_n =2}^d | \xi_{i_n,1} \cdots \xi_{i_1,1} f|^2 \geq 0. 
\end{equation}

2) This is \Cref{lem:D_products} below in the special case $m=0$.

3) We assume $g$ is in $\bigoplus_{j=0}^{n-1} \irrep_{j,j}$; the argument for $f$ is analogous. We have 
\begin{equation}
\int_{\mathbb{CP}^{d-1}} \sum_{i_1,\ldots,i_n=2}^d |\xi_{i_n,1} \cdots \xi_{i_1,1}g|^2 = \braket{g}{p_n(-\tfrac14 \Delta) g} =0.
\end{equation}
Since all terms $|\xi_{i_n,1} \cdots \xi_{i_1,1}g|^2$ are non-negative, we have $\xi_{i_n,1} \cdots \xi_{i_1,1} g=0$ for every choice of indices $i_1,\dots,i_n$. Equality $f \star_n g=0$ follows from the definition. 

4), 5) See \Cref{app:xi_appendix}.
\end{proof}

The following expression, meant as a formal (not necessarily convergent) sum,
\begin{equation}
    f \star g = \sum_{n=0}^\infty \frac{(-1)^n (m+d-1)!}{n! (n+m+d-1)!} f \star_n g,
    \label{eq:star_formal}
\end{equation}
provides an asymptotic expansion of an upper symbol of the product $\Op_m[f] \Op_m[g]$, at~least for sufficiently regular $f$ and $g$. We highlight that the only dependence on $m$ is in the coefficient $\frac{(m+d-1)!}{(n+m+d-1)!} \leq (m+d)^{-n}$. Slightly informally,
\begin{equation}
    \Op_m[f] \Op_m[g] = \Op_m[f \star g ] .
\label{eq:Op_product_formal}
\end{equation}
This was formulated precisely in \Cref{thm:2} in the Introduction, where the formal sum in \eqref{eq:star_formal} is treated by truncating to a finite order and estimating the remainder. We prove \Cref{thm:2} below. The following corollary is an immediate consequence of \Cref{thm:2} using H\"older's inequality for Schatten norms:

\begin{corollary}\label{cor:full_expansion_opf_opg}
Fix $m, N \in\N_0$ and let $p,q,r \in [1,\infty]$ satisfy $\frac{1}{r} = \frac{1}{p} + \frac{1}{q}$. Suppose that $f,g$ are functions with derivatives up to $N$th order in $L^{2p}$ and $L^{2q}$, respectively. The remainder $\mathcal E_{m,N}$ of \eqref{eq:emn_error_term_def} satisfies the bound:
    \begin{align}
    \norm{\mathcal{E}_{m,N}[f,g]}_{r} \leq
        d_m^{\frac{1}{r}} \frac{(m+d-1)!}{N!(N+m+d-1)!} \big\|f\star_N \bar f\big\|^{\tfrac12}_p\big\|\bar g \star_N g\big\|_q^{\tfrac12}.
    \end{align}
\end{corollary}

Let us note the following about \Cref{thm:2} and \Cref{cor:full_expansion_opf_opg}.
\begin{itemize}
    \item One can estimate other spectral properties of the remainder term $\mathcal E_{m,N}[f,g]$ using~2), majorization for $\Op_m$ discussed in \Cref{sec:Op_defs}, and well-known inequalities on eigenvalues and singular values of products of operators; see e.g.\ the review \cite{Bhatia}.
    \item The trace of the positive operator $\Op_m[\overline g \star_N g]$ featuring in our estimates can be computed using \eqref{eq:star_n_integral} in \Cref{prop:starn_properties}:
    \begin{equation}
        \Tr (\Op_m[\overline g \star_N g]) = d_m \int_{\CP^{d-1}} \overline g \star_N g \, \dz = d_m \braket{g}{p_n(-\tfrac14 \Delta) g}.
    \end{equation}
    \item The expression $\Op_m[f \star_N g]$ makes sense because under the running assumptions $f \star_N g $ is an integrable function.
\end{itemize}

We now turn to the proof of \Cref{thm:2}, but it will require us to first prove a few lemmas. After the proof of \Cref{thm:2}, we~will briefly describe the algebra $\mathcal{A}$, discussed in the introduction, for which \eqref{eq:star_formal} has only finitely many nonzero terms and \eqref{eq:Op_product_formal} holds as an identity.

It will be useful to consider the following operator on $L^2(\CP^{d-1},m)$:
\begin{align}
    D_m = \tfrac{1}{2}(Q-(1-\tfrac{1}{d})m^2-m(d-1)).
\end{align}
In accordance with \eqref{eq:Casimir_values}, $D_m$ is equivariant and has eigenvalues
\begin{equation}
    \mu_{m,n} = n(n+m+d-1).
\end{equation}
Its spectral decomposition is:
\begin{align}
    D_m =  \sum_{n=0}^\infty \mu_{m,n}\Pi_{n,n+m},
\end{align}
where $\Pi_{n,n+m}$ denotes the projection onto $\irrep_{n,n+m}$ in $L^2(\CP^{d-1},m)$. In particular:
\begin{align}\label{eq:projection_as_indicator_on_dm}
    V_m^{\vphantom{*}} V_m^* = \Pi_{0,m} = \mathbbm{1}_{\iset{0}}(D_m).
\end{align}
Recall that $V_m$ is the embedding introduced in \eqref{eq:Vm_def}, and that it allows to express $\Op_m[f]$ as in \eqref{eq:Op_Toeplitz}.

We will use interpolating polynomials to approximate the indicator $\mathbbm{1}_{\iset{0}}$ in \eqref{eq:projection_as_indicator_on_dm}.

\begin{definition}
    For $n\in \N_0$ we define:
    \begin{align}
    q_{m,n}(x) = \prod_{i=1}^n \frac{\mu_{m,i}-x}{\mu_{m,i}}.
    \label{eq:qmn_defined}
    \end{align}
    It is the unique polynomial such that $q_{m,n}(0)=1$ and $q_{m,n}(\mu_{m,i})=0$ for $i\in\iset{1,\ldots, n}$.
\end{definition}

\begin{lemma} \label{lem:interpolating}
    For each $n\in\N_0$ we have the formula:
    \begin{align}\label{eq:formula_q_mn_as_sum_of_prods}
        q_{m,n}(x) = \sum_{j=0}^{n}(-1)^j \prod_{i=1}^j\frac{x-\mu_{m,i-1}}{\mu_{m,i}}.
    \end{align}
    Furthermore, we have the inequalities
    \begin{subequations}
    \begin{align}
     0 \leq &(-1)^n q_{m,n}(x), \quad &&\text{for all } x \geq \mu_{m,n}, \label{eq:qmn_ineq1} \\  
    & (-1)^n q_{m,n}(x) \leq \prod_{i=1}^{n+1}\frac{x-\mu_{m,i-1}}{\mu_{m,i}}, \quad &&\text{for all } x \geq \mu_{m,n+1}. \label{eq:qmn_ineq_2}
    \end{align}
    \end{subequations}
\end{lemma}
\begin{proof}
    The formula \eqref{eq:formula_q_mn_as_sum_of_prods} follows by a simple inductive argument, noting that:
    \begin{align}
        q_{m,n+1}(x) = q_{m,n}(x)-q_{m,n}(\mu_{m,n+1})\prod_{i=1}^{n+1}\frac{x-\mu_{m,i-1}}{\mu_{m,n+1}-\mu_{m,i-1}},
    \end{align} 
    where we evaluate $q_{m,n}(\mu_{m,n+1})$ using \eqref{eq:qmn_defined}.
    
    The inequality \eqref{eq:qmn_ineq1} follows from \eqref{eq:qmn_defined}, as $(-1)^n q_{m,n}$ has a positive leading term and all zeros in $(-\infty,\mu_{m,n}]$. Using \eqref{eq:formula_q_mn_as_sum_of_prods} we can rewrite \eqref{eq:qmn_ineq_2} in the form
    \begin{equation}
        (-1)^n q_{m,n}(x) \leq (-1)^{n+1} (q_{m,n+1}(x)-q_{m,n}(x)),
    \end{equation}
    which follows from \eqref{eq:qmn_ineq1} for $q_{m,n+1}$.
\end{proof}

\begin{lemma} \label{lem:D_products}
    For each $n\in\N_0$ we have the identity:
    \begin{align}\label{eq:product_dm_minus_mui_is_star}
        \prod_{i=0}^{n-1} (D_m-\mu_{m,i}) = \sum_{i_1,\ldots, i_n = 2}^d \xi_{1,i_1}\cdots \xi_{1,i_n}\xi_{i_n,1}\cdots \xi_{i_1,1}
    \end{align}
    on $L^2(\CP^{d-1},m)$.
\end{lemma}
\begin{proof}
    Let us introduce the shorthand notation:
    \begin{align}
        \xi_{1,I} = \xi_{1,i_1}\dots \xi_{1,i_n},\quad  \xi_{I,1} = \xi_{i_1,1}\dots \xi_{i_n,1}, \quad \textup{for }I = (i_1,\ldots, i_n).
    \end{align}
    We will also use $I_{>k}$ (or $I_{<k}$) to denote the elements of $I$ preceding (or following) the $k$-th entry. We will also use $I-\iset{j}$ and $I+ \iset{j}$, where $I-\iset{j}$ is the multi-index $I$ where the first occurrence of $j$ has been removed, and $I+\iset{j}$ is the multi-index $I$ with an extra index $j$ appended. 

    We will proceed by induction. The base case is $n=1$ and is clear since $\mu_{m,0}=0$. For the inductive step, we assume the expansion holds for all $n\leq N$ and we consider:
    \begin{align}\label{eq:Dm_interpolation_inductive_step}
        D_m\prod_{i=0}^{N-1} (D_m-\mu_i) &=\sum_{j=2}^d \xi_{1,j}^L\xi_{j,1}^L \sum_{I\in \iset{2,\ldots, d}^N} \xi_{1,I}\ \xi_{I,1}\\
        &=\sum_{I'\in \iset{2,\ldots, d}^{N+1}} \xi_{1,I'}\xi_{I',1} + \sum_{I\in\iset{2,\ldots, d}^N}\sum_{j=2}^d \xi_{1,j}[\xi_{j,1},\xi_{1,I}]\xi_{I,1}.\nonumber
    \end{align}
    We note, by the Leibniz property of commutators:
    \begin{align}\label{eq:Dm_interpolation_leibniz_first_term}
        [\xi_{j,1},\xi_{1,I}] = \sum_{k=1}^N \xi_{1,I_{<k}}[\xi_{j,1}, \xi_{1,i_k}]\xi_{1,I_{>k}} &= \sum_{k=1}^N \xi_{1,I_{<k}} (\delta_{j,i_k}\xi_{1,1} -\xi_{j,i_k})   \xi_{1,I_{>k}},
    \end{align}
    where we remind that the vector fields $\xi_{i,j}$ satisfy the commutation relations \eqref{eq:xi_commutation}.
    Inserting \eqref{eq:Dm_interpolation_leibniz_first_term} into \eqref{eq:Dm_interpolation_inductive_step} we get:
    \begin{align}\label{eq:Dm_interpolation_first_manipulation}
        \sum_{I\in \iset{2,\ldots, d}^{N+1}} \xi_{1,I}\ \xi_{I,1} &+ \sum_{I\in \iset{2,\ldots, d}^{N}}\sum_{j=2}^d\sum_{k=1}^N \xi_{1,j}\xi_{1,I_{<k}} (\delta_{j,i_k}\xi_{1,1} -\xi_{j,i_k})   \xi_{1,I_{>k}}\xi_{I,1}.
    \end{align}
    The first sum is of the correct form. The second sum can be expanded as:
    \begin{align}\label{eq:Dm_interpolation_many_commutators}
        \sum_{I,j,k}\delta_{j,i_k}&\xi_{1,j}\xi_{1,I_{<k}}\xi_{1,I_{>k}}\xi_{I,1}\xi_{1,1} + \sum_{I,j,k}\delta_{j,i_k}\xi_{1,j}\xi_{1,I_{<k}}[\xi_{1,1},\xi_{1,I_{>k}}\xi_{I,1}]\\
        &-\sum_{I,j,k}\xi_{1,j} \xi_{1,I_{<k}}\xi_{I_{>k}}\xi_{I,1}\xi_{j,i_k}-\sum_{I,j,k}\xi_{1,j}\xi_{1,I_{<k}}[\xi_{j,i_k}, \xi_{1,I_{>k}}\xi_{I,1}]\nonumber\\
        =&Nm\sum_{I}\xi_{1,I}\xi_{I,1}+\sum_{I,k}\xi_{1,I_{\leq k}}[\xi_{1,1},\xi_{1,I_{>k}}\xi_{I,1}]-\sum_{I,j,k}\xi_{1,I_{<k}+\iset{j}}[\xi_{j,i_k}, \xi_{1,I_{>k}}\xi_{I,1}].\nonumber
    \end{align}
    Next, we use the identities:
    \begin{subequations}        
    \begin{align}
        [\xi_{1,1}, \xi_{1,I_{>k}}\xi_{I,1}] &= k\xi_{1,I_{>k}}\xi_{I,1},\\
        [\xi_{j,i_k}, \xi_{1,I_{>k}}\xi_{I,1}] &= \mult(j,I_{>k}) \xi_{1,I_{>k}-\iset{j}+\iset{i_k}}\xi_{I,1} \\
        &\qquad- \mult(i_k,I)\xi_{1,I_{>k}}\xi_{I-\iset{i_k}+\iset{j},1},\nonumber
    \end{align}
    \end{subequations}
    which again follow from the Leibniz property of the commutators. Here $\mult(a,A)$ denotes the number of occurrences of $a$ within $A$. Substituting into \eqref{eq:Dm_interpolation_many_commutators} we obtain:
    \begin{align}\label{eq:Dm_interpolation_many_more_commutators}
        Nm\sum_{I}\xi_{1,I}\xi_{I,1}&+\sum_{I,k}\xi_{1,I_{\leq k}}k\xi_{1,I_{>k}}\xi_{I,1}-\sum_{I,j,k}\xi_{1,I_{<k}+\iset{j}}\mult(j,I_{>k}) \xi_{1,I_{>k}-\iset{j}+\iset{i_k}}\xi_{I,1} \\
        &+ \sum_{I,j,k}\xi_{1,I_{<k}+\iset{j}} \mult(i_k,I)\xi_{1,I_{>k}}\xi_{I-\iset{i_k}+\iset{j},1}\nonumber\\
        =&Nm\sum_{I}\xi_{1,I}\xi_{I,1}+\tfrac{1}{2}N(N+1)\sum_{I}\xi_{1,I}\xi_{I,1} - \sum_{I,k}(N-k)\sum_{I}\xi_{1,I}\xi_{I,1}\nonumber\\
        &+\sum_{I,j,k}\mult(i_k,I)\xi_{1,I_+\iset{j}-\iset{i_k}} \xi_{I-\iset{i_k}+\iset{j},1}\nonumber\\
        =&N(m+1)\sum_{I}\xi_{1,I}\xi_{I,1} + \sum_{I,j,k}\mult(i_k,I)\xi_{1,I+\iset{j}-\iset{i_k}} \xi_{I-\iset{i_k}+\iset{j},1}.\nonumber
    \end{align}
    Finally we note the identity:
    \begin{align}
        \sum_{I,j,k}\mult(i_k,I)\xi_{1,I+\iset{j}-\iset{i_k}} &\xi_{I-\iset{i_k}+\iset{j},1} = \sum_{j,k}\sum_{J\in\iset{2,\ldots, d}^{N-1}}(N+d-2)\xi_{1,J+\iset{j}} \xi_{J+\iset{j},1} \nonumber \\
        &=N(N+d-2) \sum_{I\in \iset{2,\ldots, d}^N}\xi_{1,I}\xi_{I,1}.
    \end{align}
    Combining this with \eqref{eq:Dm_interpolation_first_manipulation} and \eqref{eq:Dm_interpolation_many_more_commutators} gives:
    \begin{align}
        D_m\prod_{i=0}^{N-1} (D_m-\mu_{m,i}) = \sum_{I\in \iset{2,\ldots, d}^{N+1}}\xi_{1,I}\xi_{I,1} + N(N+m+d-1)\sum_{I\in\iset{2,\ldots, d}^N}\xi_{1,I}\xi_{I,1}.
    \end{align}
    The result follows since $\mu_{m,N}=N(N+m+d-1)$.
\end{proof}

\begin{proof}[Proof of \Cref{thm:2}]
Instead of operators $\Op_m[f] = V_m^* f V_m^{\vphantom{*}}$, it is convenient to work with $\Pi_{0,m} f \Pi_{0,m} = V_m^{\vphantom{*}} \Op_m[f] V_m^*$. Under the identification of $\Sym^m(\C^d)$ with the image of $V_m$ and decomposition $L^2(\CP^{d-1},m) = \Sym^m(\C^d) \oplus \Sym^m(\C^d)^\perp$, the block form  of $\Pi_{0,m} f \Pi_{0,m}$ is:
\begin{equation}
    \Pi_{0,m} f \Pi_{0,m} = \begin{bmatrix}
        \Op_m[f] & 0 \\ 0 & 0
    \end{bmatrix}.
\end{equation}
By standard density arguments we can also assume without loss of generality that $f, g$ are smooth functions.

Consider the product 
\begin{align}
     \Pi_{0,m}\, f\,&\Pi_{0,m}\,g\, \Pi_{0,m} = \Pi_{0,m}\, f\, \mathbbm 1_{\{ 0 \}} (D_m) \,g\, \Pi_{0,m} \label{eq:Toeplitz_product_dec} \\
    = &  \Pi_{0,m}\, f \,q_{m,N-1}(D_m) \,g\, \Pi_{0,m} + \Pi_{0,m}\, f \,[\mathbbm 1_{\{ 0 \}} (D_m) - q_{m,N-1}(D_m) ]\, g\, \Pi_{0,m}, \nonumber
\end{align}
where we split $\mathbbm 1_{\{ 0 \}} (D_m)$ into the main term, given by an interpolating polynomial, and a~remainder. For the main term we write, using \Cref{lem:interpolating}:
\begin{equation}
      \Pi_{0,m}\, f\, q_{m,N-1}(D_m) \,g\, \Pi_{0,m} = \sum_{n=0}^{N-1} (-1)^n  \Pi_{0,m} \, f  \left[ \prod_{i=1}^n \frac{D_m - \mu_{m,i-1}}{\mu_{m,i}} \right] g \, \Pi_{0,m}.
\end{equation}
The product $\prod_{i=1}^n (D_m - \mu_{m,i-1})$ can be re-expressed using \eqref{eq:product_dm_minus_mui_is_star}, since multiplication by smooth functions on $\CP^{d-1}$ does not map out of $L^2(\CP^{d-1},m)$. We have $\xi_{i,1} \Pi_{0,m} =0$ and $\Pi_{0,m} \xi_{1,i} =0$ (see \Cref{prop:sections_decomp}), so we can commute the $\xi$ operators with $f,g$ and in effect get them to act on $f$ and $g$ only. This leads to the identity:
\begin{equation}
    \Pi_{0,m}\,f \left[ \prod_{i=1}^n (D_m - \mu_{m,i-1}) \right] g \,\Pi_{0,m} = \Pi_{0,m}\, (f \star_n g) \,\Pi_{0,m}.
\end{equation}
Let us illustrate this simple derivation in the case $n=1$:
\begin{equation}
\Pi_{0,m}\,f \xi_{1,i} \xi_{i,1} g \,\Pi_{0,m} = \Pi_{0,m} \, [f,\xi_{1,i}] [\xi_{i,1},g] \, \Pi_{0,m},
\end{equation}
where $[\cdot , \cdot ]$ is the commutator. Now note that $[\xi_{i,1},g]$ and $[f,\xi_{1,i}]$ are operators of multiplication by $\xi_{i,1}(g)$ and $- \xi_{1,i}(f)$, respectively. The case of general $n$ differs only in the amount of required bookkeeping. 

We have shown that
\begin{equation}
    \Pi_{0,m}\, f \,q_{m,N-1}(D_m) \,g\, \Pi_{0,m} = \sum_{n=0}^{N-1} (-1)^n \left[ \prod_{i=1}^n \frac{1}{\mu_{m,i}} \right] \Pi_{0,m} \, (f \star_n g) \, \Pi_{0,m}.
\end{equation}

We have to bound the remainder term in \eqref{eq:Toeplitz_product_dec}. The inequalities in Lemma \ref{lem:interpolating}, together with $\mathbbm 1_{\{ 0 \}}(\mu_{m,i}) - q_{m,N-1}(\mu_{m,i}) =0$ for every $i \leq N-1$, imply that
\begin{equation}
 0 \leq  (-1)^N [\mathbbm 1_{\{ 0 \}} (D_m) - q_{m,N-1}(D_m)]  \leq \prod_{i=1}^{N} \frac{D_m - \mu_{m,i-1}}{\mu_{m,i}},
\end{equation}
and hence for $g= \overline f$:
\begin{align}
    0 &\leq (-1)^N \Pi_{0,m} \, f [  \mathbbm 1_{\{ 0 \}} (D_m) - q_{m,N-1}(D_m) ] \overline f \, \Pi_{0,m} \\
    &\leq \Pi_{0,m} \, f \left[ \prod_{i=1}^{N} \frac{D_m - \mu_{m,i-1}}{\mu_{m,i}} \right] \overline f \, \Pi_{0,m} . \nonumber
\end{align}
The right hand side can now be manipulated using the same steps as for the main terms. This establishes 1). 

To verify 2), we set
\begin{subequations}
\begin{align}
    A &= \sqrt{\frac{N! (N+m+d-1)!}{(m+d-1)!}} \sqrt{(-1)^N [ \mathbbm 1_{\{ 0 \}}(D_m) -q_{m,N-1}(D_m)]} \, \overline f \, \Pi_{0,m}, \\
    B &= \sqrt{\frac{N! (N+m+d-1)!}{(m+d-1)!}} \sqrt{(-1)^N [ \mathbbm 1_{\{ 0 \}}(D_m) -q_{m,N-1}(D_m)]} \, g \, \Pi_{0,m}.
\end{align}
\end{subequations}
Then we can bound $A^*A$ and $B^* B$ as in the proof of 1). Point 3) follows from H\"older's inequality for operators. 
\end{proof}

We will call $f \in L^2(\CP^{d-1})$ regular if the linear span of the $\SU(d)$-orbit of $f$ is finite-dimensional. We~denote the algebra of regular functions on $\CP^{d-1}$ by $\mathcal A$. It is easy to see that:
\begin{equation}
    \mathcal A = \bigcup_{n=0}^\infty \bigoplus_{j=0}^n \irrep_{j,j}. 
\end{equation}
For every two $f,g \in \mathcal A$, we define their star product $f \star g$ as in \eqref{eq:star_formal}. The sum has only finitely many nonzero terms by 3. in \Cref{prop:starn_properties}, and every $f \star_n g$ is a regular function. Therefore, $f \star g$ defines a regular function for every complex number $m$ for which the coefficients $\frac{(m+d-1)!}{(n+m+d-1)!}$ are finite, i.e.\ for $m \not \in \{ -d ,-d-1,\dots \}$. The identity \eqref{eq:Op_product_formal} is satisfied by \Cref{thm:2} with sufficiently large $N$ (dependent on $f$ and $g$).

\begin{proposition}
    $\star$ is associative on $\mathcal A$ for every complex $m \not \in \{ -d , -d+1,\dots \}$. 
\end{proposition}
\begin{proof}
Let $f,g,h \in \mathcal A$. Evaluating $(f \star g) \star h$ and $f \star (g \star h)$ using the definition \eqref{eq:star_formal} we find linearly independent regular functions $\varphi_1,\dots,\varphi_k $ and rational coefficient functions $c_1(m),\dots,c_k(m)$ such that
\begin{equation}
    (f \star g ) \star h - f \star (g \star h) = \sum_{i=1}^k c_i (m) \varphi_i.
    \label{eq:star_associator}
\end{equation}
Since $c_i(m)$ are rational, we will be able to conclude that they all vanish by verifying that $(f \star g ) \star h - f \star (g \star h)$ vanishes for infinitely many values of $m$. 

By~construction, for~every $m \in \mathbb N_0$ we have
\begin{equation}
    \Op_m[(f \star g) \star h] = \Op_m[f \star g ] \Op_m[h]= (\Op_m[f]\Op_m[g]) \Op_m[h].
\end{equation}
Repeating similar steps for $\Op_m[f \star (g \star h)]$ and using associativity of composition of linear operators we find
\begin{equation}
    \Op_m[(f \star g) \star h] =\Op_m[f \star (g \star h)].
    \label{eq:opm_associative}
\end{equation}
For all sufficiently large integers $m$ the operators $\Op_m[\varphi_i]$ are linearly independent, and~for such $m$ we have $c_i(m)=0$ by comparing \eqref{eq:star_associator} with \eqref{eq:opm_associative}.
\end{proof}

\appendix
\renewcommand{\thesection}{\Alph{section}}
\setcounter{definition}{0}
\renewcommand{\thedefinition}{\Alph{section}\arabic{definition}}
\setcounter{equation}{0}
\renewcommand{\theequation}{\Alph{section}\arabic{equation}}

\section{Geometric structures on complex projective space} \label{app:xi_appendix}

In this Appendix we describe geometric structures on $\CP^{d-1}$ in terms of left-invariant vector fields on $\SU(d)$. The presentation builds on Subsection~\ref{sec:homog}, and we assume the reader is familiar with the definitions and notations introduced there. In the discussion of affine connections we follow the ideas of Nomizu \cite{Nomizu-connections}.

For $g \in \SU(d)$, we decompose the tangent space of $\SU(d)$ at $g$ as
\begin{equation}
    \mathrm{T}_g \SU(d) = V_g \oplus H_g,
\end{equation}
where the \emph{vertical subspace} $V_g$ is the tangent space to the coset $g \PStab(e_1)$, and the \emph{horizontal subspace} $H_g$ is the orthogonal complement of $V_g$ with respect to the bi-invariant Riemannian metric on $\SU(d)$. The complexification of $H_g$ has a basis consisting of the values at $g$ of the vector fields $\{ \xi_{1,i}, \xi_{i,1} \}_{i=2}^d$.

If $X$ is a vector field on $\SU(d)$ invariant under the right action of $\PStab(e_1)$, the pushforward of $X$ by $\pi : \SU(d) \to \CP^{d-1}$ is well-defined. The pushforward map $d \pi$ annihilates vertical vector fields, and establishes a $C^\infty(\CP^{d-1})$-linear bijection between horizontal, $\PStab(e_1)$-invariant vector fields on $\SU(d)$ and vector fields on $\CP^{d-1}$.

The K\"ahler structure of $\CP^{d-1}$ can be described neatly using the above identification: for every $g \in G$, the vectors $\left. d \pi \right|_g \xi_{i,1}$ form an orthonormal basis of the holomorphic tangent space of $\CP^{d-1}$ at $\pi(g)$. In particular, the Riemannian scalar product of gradients of two functions $f,g$ on $\CP^{d-1}$ is given by:
\begin{equation}
\nabla f \cdot \nabla g =     -2 \sum_{i=2}^d \left( \xi_{i,1}(f) \xi_{1,i}(g) + \xi_{1,i}(f) \xi_{i,1}(g)  \right),
\end{equation}
and their Poisson bracket is given by:
\begin{equation}
\{ f, g \} = 2i \sum_{i=2}^d \left( \xi_{i,1}(f) \xi_{1,i}(g) - \xi_{1,i}(f) \xi_{i,1}(g)  \right).
\end{equation}

If $X,Y$ are horizontal, $\PStab(e_1)$-invariant vector fields, then
\begin{equation}
    [d \pi (X), d \pi (Y)] = d \pi [X,Y]. 
    \label{eq:dpi_commutator}
\end{equation}
In general $[X,Y]$ may have a vertical component, which is annihilated by $d \pi$ in \eqref{eq:dpi_commutator}. Therefore, the Lie bracket of vector fields on $\CP^{d-1}$ is obtained from the Lie bracket of horizontal, $\PStab(e_1)$-invariant vector fields by projecting onto the horizontal component.

Let $\nabla^{(-)}$ be the Cartan's (-)-connection on $\SU(d)$; that is, the unique affine connection for which all left-invariant vector fields are parallel. It satisfies two key properties:
\begin{itemize}
    \item If $X,Y$ are vector fields with $Y$ horizontal, then $\nabla^{(-)}_X Y$ is also horizontal,
    \item If $X,Y$ are $\PStab(e_1)$-invariant, then so is $\nabla_X^{(-)} Y$.
\end{itemize}
Consequently, if $ X', Y'$ are vector fields on $\CP^{d-1}$, identified with their horizontal, $\PStab(e_1)$-invariant lifts $X,Y$ on $\SU(d)$, then $\nabla_X^{(-)} Y$ is again horizontal and $\PStab(e_1)$-invariant. We denote the corresponding vector field on $\CP^{d-1}$ by $\nabla_{X'}Y'$. This construction yields a metric-compatible affine connection $\nabla$ on $\CP^{d-1}$.

We will now show that $\nabla$ is torsion-free and therefore coincides with the Levi-Civita connection. For vector fields $X,Y$ and $X',Y'$ as above, we have
\begin{equation}
    \nabla_{X'} Y' - \nabla_{Y'} X' - [X',Y'] = d \pi \left( \nabla^{(-)}_X Y - \nabla_Y^{(-)} X - [X,Y] \right).  
\end{equation}
To evaluate this, we expand $X,Y$ in the basis of left-invariant horizontal vector fields $\xi_{i,1}$ and $\xi_{1,i}$. The terms involving derivatives of the coefficient functions cancel, leaving only expressions proportional to $d \pi [\xi_{i,1},\xi_{1,j}]$. These vanish because the commutators $[\xi_{i,1},\xi_{1,j}]$ are vertical.

Next, we calculate the curvature of $\nabla$. We consider vector fields $X',Y',Z'$ on $\CP^{d-1}$ with lifts $X,Y,Z$. Denoting the horizontal part of $[X,Y]$ by $[X,Y]_H$,
\begin{align}
\mathrm{curv}(X',Y')Z' &=  (  \nabla_{X'} \nabla_{Y'} - \nabla_{Y'} \nabla_{X'} - \nabla_{[X',Y']}) Z'  \\
&= d \pi \left( \nabla_X^{(-)} \nabla_Y^{(-)} - \nabla_Y^{(-)} \nabla_X^{(-)} - \nabla_{[X,Y]_H}^{(-)} Z \right). \nonumber
\end{align}
The whole parenthesis would vanish if $[X,Y]_H$ was replaced by $[X,Y]$ because $\nabla^{(-)}$ is a flat connection. Hence:
\begin{equation}
      \mathrm{curv}(X',Y')Z' = d \pi \left( \nabla_{[X,Y]_V}^{(-)} Z \right),
\end{equation}
where $[X,Y]_V$ is the vertical part of $[X,Y]$. Let us represent this more explicitly. We~choose a basis $e_i$ of left-invariant horizontal fields\footnote{One can take $\xi_{i,1}$ and $\xi_{1,i}$; here it is convenient to use one symbol for all basis elements.} and decompose $X = \sum_i X^i e_i$, and analogously for $Y$ and $Z$. Then
\begin{equation}
    [X,Y]_V = \sum_{i,j} X^i Y^j [e_i,e_j]_V.
    \label{eq:curvature_intermediate}
\end{equation}
Terms with derivatives of the coefficients functions $X^i,Y^j$ are absent because they give rise to horizontal contributions. Next,
\begin{equation}
    \nabla_{[X,Y]_V}^{(-)} Z = \sum_{i,j,k} X^i Y^j [e_i,e_j]_V(Z^k) e_k.
\end{equation}
To get rid of derivatives acting on $Z^k$, we use the invariance of $Z$ under right translations by $\PStab(e_1)$. Since $[e_i,e_j]_V$ is left-invariant and vertical, it is among the generators of this group action, so:
\begin{equation}
0 = [[e_i,e_j]_V ,Z] = \sum_k \left( [e_i,e_j]_V (Z^k) e_k + Z^k [[e_i,e_j]_V,e_k] \right). 
\end{equation}
Hence, combining this with \eqref{eq:curvature_intermediate},
\begin{equation}
 \mathrm{curv}(X',Y')Z' = d \pi \left( - \sum_{i,j,k} X^i Y^j Z^k  [[e_i,e_j]_V,e_k] \right).
\end{equation}
We remark that the subscript $V$ is actually redundant because $[e_i,e_j]$ is vertical. The vector field in the parenthesis is horizontal and $\PStab(e_1)$-invariant. 

The curvature formulas above show directly that the curvature is a $2$-form of type $(1,1)$. Indeed, in the computation of the $(0,2)$ component we encounter commutators $[\xi_{1,i},\xi_{1,j}]_V$, which are all zero \eqref{eq:xi_commutation}.

The Levi-Civita connection is extended to tensors in the standard way, and this allows to define higher covariant derivatives by iteration. In particular, if $f$ is a function on $\CP^{d-1}$, its $k$th covariant derivative is a section $\nabla^k f$ of the $k$th power of the cotangent bundle. We denote by $\nabla^{(k,0)} f$ and $\nabla^{(0,k)} f$ its $(k,0)$ and $(0,k)$ type components. They are symmetric tensors, while $\nabla^k$ generally speaking is not if $k \geq 3$. Given two functions $f,g$, the contraction of $\nabla^{(0,k)} f$ and $\nabla^{(k,0)} g$ (defined using the $k$-fold tensor product of the Riemannian metric tensor) is
\begin{equation}
    \nabla^{(0,k)} f \cdot \nabla^{(k,0)}g = (-2)^k \sum_{i_1,\dots,i_k =2}^d \xi_{1,i_1} \cdots \xi_{1,i_k}(f) \xi_{i_1,1} \cdots \xi_{i_k,1}(g) =2^k f \star_k g,
\end{equation}
where in the last equality we used notation introduced in Section \ref{sec:star}.

\printbibliography
\end{document}